%% file: main.tex
\documentclass[dvipsnames]{llncs}

\input{macros}

\title{Symbolic Register Automata\thanks{This work was partially funded by National Science Foundation Grants CCF-1763871, CCF-1750965, a Facebook TAV Research Award, the ERC starting grant Profoundnet (679127) and a Leverhulme Prize (PLP-2016-129).}}

\author{
Loris D'Antoni\inst{1} \and
Tiago Ferreira\inst{2} \and
Matteo Sammartino\inst{2} \and
Alexandra Silva\inst{2}}

\institute{
University of Wisconsin–Madison, Madison, WI 53706-1685, USA\\
\email{loris@cs.wisc.edu} \and
University College London, Gower Street, London, WC1E 6BT, UK\\
\email{me@tiferrei.com, \{m.sammartino,a.silva\}@ucl.ac.uk}
}

\begin{document}
\maketitle

\begin{abstract}
Symbolic Finite Automata and Register Automata are two orthogonal extensions of finite automata motivated by real-world problems where data may have unbounded domains. These automata address a demand for a model over large or infinite alphabets, respectively. Both automata models have interesting applications and have been successful in their own right.  In this paper, we introduce Symbolic Register Automata, a new model that combines features from both symbolic and register automata, with a view on applications that were previously out of reach. We study their properties and provide algorithms for emptiness, inclusion and equivalence checking, together with experimental results. 
\end{abstract}

\input{intro.tex}
\input{example.tex}
\input{sra.tex}
\input{decid.tex}
\input{appl.tex}
\input{related.tex}

\bibliographystyle{abbrv}
\bibliography{cav}

\input{appendix}

\end{document}

%% file: macros.tex
\usepackage{xspace}
\usepackage{stmaryrd}
\usepackage{amsmath,amssymb}
\usepackage{stmaryrd}
\usepackage{marginnote}
\usepackage{tikz-cd}

\usepackage{wrapfig}
\usepackage{color}
\usepackage{eucal}
\usepackage{nicefrac}
\usepackage{mathtools}
\usepackage{array}
\usepackage[textsize=tiny]{todonotes}
\usepackage{mathrsfs}
\usepackage{enumitem}
\usepackage{algorithm2e}
\usepackage{enumitem}
\usepackage{multirow}

\usepackage{subfig}

\usepackage{tikz,pgf}
\usetikzlibrary{arrows, positioning, automata, calc, shapes, cd, spy, backgrounds, decorations}
\usepackage{pgfplots}

\newcommand{\alg}{\mathcal{A}}
\newcommand{\els}{\mathcal{D}}
\newcommand{\preds}{\Psi}

\newcommand{\mint}{\mathsf{mint}}
\newcommand{\inm}{\sqsubset}

\newcommand{\den}[1]{\llbracket #1 \rrbracket}
\newcommand{\sra}{\ensuremath{\mathcal{S}}}

\newcommand{\lang}{\ensuremath{\mathscr{L}}}
\newcommand{\tr}[3][]{\xrightarrow[#1]{#2/#3}}

\newcommand{\dom}{\mathsf{dom}}
\newcommand{\img}{\mathsf{img}}

\newcommand{\fresh}[1]{#1^{\bullet}}
\newcommand{\checkreg}[1]{#1^{=}}

\newcommand{\nrm}[1]{\mathsf{N}(#1)}
\newcommand{\ns}[2]{{#2} \rhd {#1}}
\newcommand{\atom}{\mathsf{atom}}
\newcommand{\fres}[2]{#1\!\restriction\!#2}
\newcommand{\Pow}{\mathcal{P}}
\newcommand{\bigO}{\mathcal{O}}

\newcommand{\clts}{\mathsf{CLTS}}

\newcommand{\sat}{\mathsf{isSat}}
\newcommand{\regab}{\mathscr{E}}
\newcommand{\card}[1]{|#1|}%

\newcolumntype{P}[1]{>{\centering\arraybackslash}m{#1}}

\newcommand{\srasim}{\mathcal{R}}
\newcommand{\ltssim}{\mathcal{R}}
\newcommand{\nrmsim}{\stackrel{\mathtt{N}}{\prec}}
\newcommand{\nrmbsim}{\stackrel{\mathtt{N}}{\sim}}

\def\SFA{SFA\xspace}
\def\SFAs{SFAs\xspace}
\def\SEFA{SEFA\xspace}
\def\SEFAs{SEFAs\xspace}
\def\RA{RA\xspace}
\def\RAs{RAs\xspace}
\def\SRA{SRA\xspace}
\def\SRAs{SRAs\xspace}

\newcommand{\rone}{(\emph{i})~}
\newcommand{\rtwo}{(\emph{ii})~}

%% file: intro.tex
\section{Introduction}

Finite automata are a ubiquitous formalism that is simple enough to model many real-life systems and phenomena. 
They enjoy a large variety of theoretical properties that in turn play a role in practical applications.
For example,  finite automata are closed under Boolean operations, and
have decidable emptiness and equivalence checking procedures.
Unfortunately, finite automata have a fundamental limitation: 
they can only operate over finite (and typically small) alphabets.
Two \emph{orthogonal} families of automata models have been proposed to overcome this:
\emph{symbolic automata} and \emph{register automata}. 
In this paper, we show that these two models can be combined yielding a new powerful model that can cover interesting applications previously out of reach for existing models. 

Symbolic finite  automata (\SFA) allow transitions to carry predicates over
rich first-order alphabet theories, such as linear arithmetic, and therefore extend
classic automata to operate over infinite alphabets~\cite{DAntoniV17}. 
For example, an \SFA can define the language of all lists of integers in which the first and last elements are positive 
integer numbers.
Despite their increased expressiveness, \SFAs enjoy the same closure and decidability properties of finite automata---e.g., closure under
Boolean operations and decidable equivalence and emptiness.

Register automata (\RA) support infinite alphabets by allowing 
input characters to be stored in registers during the computation and to be compared
against existing values that are already stored in the registers~\cite{KaminskiF94}.
For example, an \RA can define the language of all lists of integers in which all numbers appearing in even positions are the same.
\RAs do not have some of the properties of finite automata (e.g., they cannot be determinized), but they still enjoy many 
useful properties that have made them a popular model in static analysis, software verification, and program monitoring~\cite{GrigoreDPT13}.

In this paper, we combine the best features of these two models---first order alphabet theories and registers---into a new model, 
\emph{symbolic register automata} (\SRA). \SRA are strictly more expressive than \SFA and \RA. For example, an \SRA can define the language of all lists of integers in which the first and last elements are positive rational numbers
and all numbers appearing in even positions are the same. This language is not recognizable by either an \SFA nor by an \RA.

While other attempts at combining symbolic automata and registers have resulted in undecidable models with limited closure properties~\cite{DAntoni2015},
we show that \SRAs enjoy the same closure and decidability properties of (non-symbolic) register automata. We propose a new application enabled by \SRAs and implement our model in an open-source automata library. 

In summary, our contributions are:
\begin{itemize}
	\item Symbolic Register Automata (\SRA): a new automaton model that can handle complex alphabet theories while allowing symbols at arbitrary positions in the input string to be compared using equality (\S~\ref{sec:model}).
	\item A thorough study of the properties of \SRAs. We show that \SRAs are closed under intersection, union and (deterministic) complementation, and provide algorithms for emptiness and forward (bi)simulation (\S~\ref{sec:decision}).
	\item A study of the effectiveness of our \SRA implementation on handling regular expressions with back-references (\S~\ref{sec:eval}). We compile a set of benchmarks from existing regular expressions with back-references (e.g., \texttt{(\textbackslash d)[a{-}z]$^*$\textbackslash 1}) and 
show that \SRAs are an effective model for such expressions and existing models such as \SFAs and \RAs are not. Moreover, we show that \SRAs are more efficient than the \texttt{java.util.regex} library for matching regular expressions with back-references.

\end{itemize}

%% file: example.tex
\section{Motivating example}
\label{sec:mot-ex}

\begin{figure}[t]  
\centering
     \subfloat[Regular expression $r_p$ (with back-reference). \label{ex:ipregex:regex}]{%
       \texttt{C:(.\{3\}) L:(.) D:[\^{}\textbackslash s]+( C:\textbackslash 1 L:\textbackslash 2 D:[\^{}\textbackslash s]+)+ }
     }\\
\makebox[\textwidth][c]{     
     \subfloat[Example text matched by $r_{p}$.\label{ex:ipregex:match}]{%
       \texttt{     C:X4a L:4 D:bottle C:X4a L:4 D:jar}
     }\hfill
     \subfloat[Example text \emph{not} matched by $r_{p}$.\label{ex:ipregex:notmatch}]{%
       \texttt{C:X4a L:4 D:bottle C:X5a L:4 D:jar}
     }
     }     
     \subfloat[Snippets of a symbolic register automaton $A_{p}$ corresponding to $r_{p}$.\label{ex:ipregex:sra}]{%
            \includegraphics[width=\textwidth]{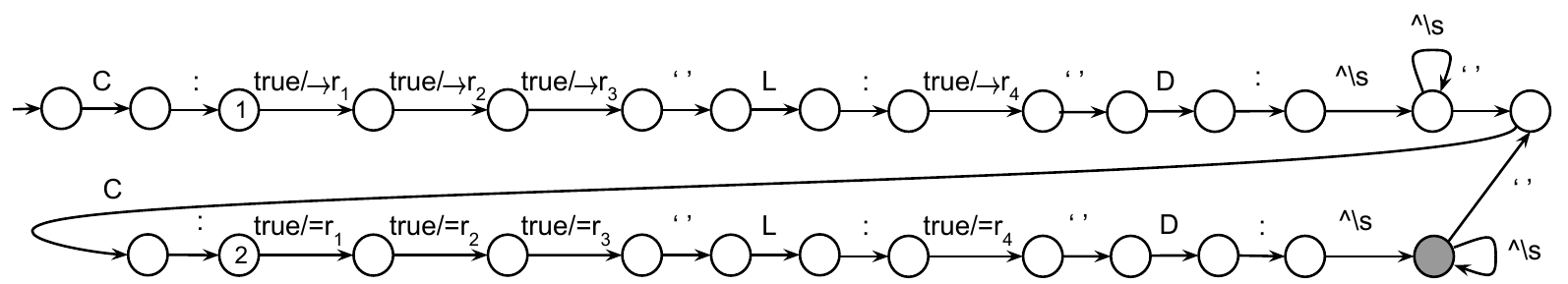}
     }    
     \vspace{-2ex}
\caption{Regular expression for matching products with same code and lot number---i.e., the characters of C and L are the same in all the products.
\label{ex:ipregex}}
\end{figure}

In this section, we illustrate the capabilities of symbolic register automata using a simple example.
Consider the regular expression $r_{p}$ shown in Figure~\ref{ex:ipregex:regex}.
This  expression, given a sequence of product descriptions, checks whether the products have the same code and lot number.
The reader might not be familiar with some of the unusual syntax of this expression. In particular, $r_{p}$ uses
two back-references \texttt{\textbackslash 1} and \texttt{\textbackslash 2}. The semantics of this construct is that the string matched by the regular expression
for \texttt{\textbackslash 1} (resp. \texttt{\textbackslash 2}) should be exactly the string that matched the subregular expression $r$ appearing between the first (resp. second) two parenthesis, in this case \texttt{(.\{3\})} (resp. \texttt{(.)}).
Back-references allow regular expressions to check whether the encountered text is the same or is different from a string/character that appeared earlier in the input (see Figures~\ref{ex:ipregex:match} and \ref{ex:ipregex:notmatch} for examples of positive and negative matches).

Representing this complex  regular expression using an automaton model requires addressing several challenges.
The expression $r_{p}$:
\begin{enumerate}[topsep=2pt,noitemsep]
\item operates over large input alphabets consisting of upwards of 2$^{16}$ characters;
\item uses complex character classes (e.g., \texttt{\textbackslash s}) to describe different sets of characters in the input;
\item adopts back-references to detect repeated strings in the input.
\end{enumerate}
Existing automata models do not address one or more of these challenges.
Finite automata require one transition for each character in the input alphabet and blow-up when representing large alphabets.
Symbolic finite  automata (\SFA) allow transitions to carry predicates over
rich structured first-order  alphabet theories and can describe, for example, character classes~\cite{DAntoniV17}. 
However, \SFAs cannot directly check whether a character or a string is repeated in the input. 
An \SFA for describing the regular expression $r_{p}$ would have to store the characters after \texttt{C:}  directly
in the states to later check whether they match the ones of the second product. Hence,
the smallest \SFA for this example would require billions of states!
Register automata (\RA) and their variants can store characters in registers during the computation and compare
characters
against values already stored in the registers~\cite{KaminskiF94}.
Hence, \RAs can check whether the two products have the same code. However, \RAs only operate over unstructured infinite
alphabets and cannot check, for example, that a character belongs to a given class.

The model we propose in this paper, \emph{symbolic register automata} (\SRA), combines the best features of \SFAs and \RAs---first-order alphabet theories and registers---and can address all the three aforementioned challenges.
Figure~\ref{ex:ipregex:sra} shows a snippet of a symbolic register automaton $A_{p}$ corresponding to $r_{p}$. 
Each transition in $A_{p}$ is labeled with a predicate that describes what characters can trigger the transition.
For example, \texttt{\^{}\textbackslash s} denotes that the transition can be triggered by any non-space character, 
\texttt{L} denotes that the transition can be triggered by the character \texttt{L},
and
\texttt{true} denotes that the transition can be triggered by any character.
Transitions of the form $\varphi/\!\!\rightarrow \! r_i$ denote that, if a character $x$ satisfies the predicate $\varphi$, the character is then stored in the register $r_i$. For example, the transition out of state 1 reads any character and stores it in register $r_1$.
Finally, transitions of the form $\varphi/\!\!=r_i$ are triggered if a character $x$ satisfies the predicate $\varphi$ and $x$ is the same character as the one stored in $r_i$.
For example, the transition out of state 2 can only be triggered by the same character that was stored in $r_1$ when reading the transition
out state 1---i.e., the first characters in the product codes should be the same.

\SRAs are a natural model for describing regular expressions like $r_{p}$, where capture groups are of bounded length, and hence correspond to finitely-many registers.
The \SRA $A_{p}$ has fewer than 50 states (vs. more than 100 billion for \SFAs) and can, for example, be used to
 check whether an input string matches the given regular expression (e.g., monitoring).
More interestingly, in this paper we study the closure and decidability properties of \SRAs and provide an implementation for
our model.
For example, consider the following regular expression $r_{pC}$ that only checks whether the product codes are the same, but not the lot numbers:
\[      
\texttt{C:(.\{3\}) L:. D:[\^{}\textbackslash s]+( C:\textbackslash 1 L:. D:[\^{}\textbackslash s]+)+}
 \]
The set of strings accepted by $r_{pC}$ is a superset of the set of strings accepted by $r_{p}$.
In this paper, we present simulation and bisimulation algorithms that can check this property.
Our implementation can show that $r_{p}$ subsumes $r_{pC}$ in 25 seconds
and we could not find other tools that can prove the same property.

%% file: sra.tex
\section{Symbolic Register Automata}
\label{sec:model}

In this section we introduce some preliminary notions, we define symbolic register automata and a variant that will be useful in proving decidability properties.

\smallskip
\noindent
\textbf{Preliminaries.} An \emph{effective Boolean algebra} $\alg$ is a tuple $(\els, \preds, \den{\_}, \bot, \top,\land,\lor,\neg)$, where: $\els$ is a set of domain elements; $\preds$ is a set of predicates closed under the Boolean connectives and $\bot, \top \in \preds$. The denotation function
$\den{\_} \colon : \preds \to 2^{\els}$ is such that $\den{\bot} = \emptyset$ and $\den{\top} = \els$, for all $\varphi,\psi \in \preds$, $\den{\varphi \lor \psi} = \den{\varphi} \cup \den{\psi}$, $\den{\varphi \land \psi} = \den{\varphi} \cap \den{\psi}$, and $\den{\neg \varphi} = \els \setminus \den{\varphi}$. For $\varphi \in \preds$, we write $\sat(\varphi)$ whenever $\den{\varphi} \neq \emptyset$ and say that $\varphi$ is \emph{satisfiable}. $\alg$ is \emph{decidable} if $\sat{}$ is decidable. For each $a \in \els$, we assume predicates $\atom(a)$ such that $\den{\atom(a)} = \{a\}$.

\begin{example}
The theory of linear integer arithmetic forms an effective BA, where $\els = \mathbb{Z}$ and $\preds$ contains formulas $\varphi(x)$ in the theory with one fixed integer variable. For example, $\mathsf{div_k} \coloneqq
(x \; \mathsf{mod} \; k) = 0$ denotes the set of all
integers divisible by $k$. 	
\end{example}

\vspace{-1ex}
\noindent{\bf Notation.} Given a set $S$, we write $\Pow(S)$ for its powerset. Given a function $f \colon A \to B$, we write $f[a \mapsto b]$ for the function such that $f[a \mapsto b](a) = b$ and $f[a \mapsto b](x) = f(x)$, for $x \neq a$. Analogously, we write $f[S \mapsto b]$, with $S \subseteq A$, to map multiple values to the same $b$. The \emph{pre-image} of $f$ is the function $f^{-1} \colon \Pow(B) \to \Pow(A)$ given by $f^{-1}(S) = \{ a \mid \exists b \in S \colon b = f(a) \}$; for readability, we will write $f^{-1}(x)$ when $S = \{x\}$. Given a relation $\mathcal{R} \subseteq A \times B$, we write $a \mathcal{R} b$ for $(a,b) \in \mathcal{R}$.

\medskip
\noindent\textbf{Model definition.} 
Symbolic register automata have transitions of the form:
\[
	p \tr{\varphi}{E,I,U} q
\]
where $p$ and $q$ are states, $\varphi$ is a predicate from a fixed effective Boolean algebra, and $E,I,U$ are subsets of a fixed finite set of registers $R$. The intended interpretation of the above transition is: an input character $a$ can be read in state $q$ if (i) $a \in \den{\varphi}$, (ii) the content of all the registers in $E$ is \emph{equal} to $a$, and (iii) the content of all the registers in $I$ is \emph{different} from $a$. If the transition succeeds then $a$ is stored into all the registers $U$ and the automaton moves to $q$.
\begin{example}
The transition labels in Figure~\ref{ex:ipregex:sra} have been conveniently simplified to ease intuition. These labels correspond to full SRA labels as follows:
\[
	\varphi / \!\! \rightarrow \!\! r \; \Longrightarrow \; \varphi / \emptyset,\emptyset,\{r\} \qquad \varphi / \! = \!r \; \Longrightarrow \; \varphi / \{r\},\emptyset,\emptyset \qquad \varphi \; \Longrightarrow \; \varphi / \emptyset,\emptyset,\emptyset \enspace .
\]
\end{example}
Given a set of registers $R$, the transitions of an \SRA have labels over the following set: $L _R= \preds \times \{ (E,I,U) \in \Pow(R) \times \Pow(R) \times \Pow(R) \mid E \cap I = \emptyset \}$. The condition $E \cap I = \emptyset$ guarantees that register constraints are always satisfiable.
\begin{definition}[Symbolic Register Automaton] A \emph{symbolic register automaton} (\SRA) is a 6-tuple $(R,Q,q_0,v_0,F,\Delta)$, where
$R$ is a finite set of \emph{registers}, $Q$ is a finite set of \emph{states}, $q_0 \in Q$ is the \emph{initial} state,
$v_0 \colon R \to \els \cup \{\sharp\}$ is the \emph{initial register assignment} (if $v_0(r) = \sharp$, the register $r$ is considered \emph{empty}),
$F \subseteq Q$ is a finite set of \emph{final} states, and $\Delta \subseteq Q \times L_R \times Q$ is the \emph{transition relation}. 
Transitions $(p,(\varphi,\ell),q) \in \Delta$ will be written as $p \tr{\varphi}{\ell} q$.
\end{definition}
An SRA can be seen as a finite description of a (possibly infinite) labeled transition system (LTS), where states have been assigned concrete register values, and transitions read a single symbol from the potentially infinite alphabet. This so-called \emph{configuration LTS} will be used in defining the semantics of SRAs.

\begin{definition}[Configuration LTS]
Given an \SRA $\sra$, the \emph{configuration} LTS $\clts(\sra)$ is defined as follows.
A \emph{configuration} is a pair $(p,v)$ where $p \in Q$ is a state in $\sra$ and a $v \colon R \to \els \cup \{\sharp\}$ is \emph{register assignment};
$(q_0,v_0)$ is called the \emph{initial configuration}; every $(q,v)$ such that $q \in F$ is a \emph{final} configuration. 
The set of transitions between configurations is defined as follows:
\[
\frac{p \tr{\varphi}{E,I,U} q \in \Delta \qquad E \subseteq v^{-1}(a) \quad I \cap v^{-1}(a) = \emptyset
}{
(p,v) \xrightarrow{a} (q,v[U \mapsto a]) \in \clts(\sra) 
}
\]
\end{definition}	
Intuitively, the rule says that a SRA transition from $p$ can be instantiated to one from $(p,v)$ that reads $a$ when the registers containing the value $a$, namely $v^{-1}(a)$, satisfy the constraint described by $E,I$ ($a$ is contained in registers $E$ but not in $I$). If the constraint is satisfied, all registers in $U$ are assigned $a$.

A \emph{run} of the \SRA $\sra$ is a sequence of transitions in $\clts(\sra)$ starting from the initial configuration. A configuration is \emph{reachable} whenever there is a run ending up in that configuration. The \emph{language} of an SRA $\sra$ is defined as 
\[	
	\lang(\sra) := \{ a_1 \dots a_n \in \els^n \mid \exists (q_0,v_0) \xrightarrow{a_1} \dots \xrightarrow{a_n} (q_n,v_n) \in \clts(\sra), q_n \in F \}
\]
An SRA $\sra$ is \emph{deterministic} if its configuration LTS is; namely, for every word $w \in \els^\star$ there is at most one run in $\clts(\sra)$ spelling $w$. Determinism is important for some application contexts, e.g., for runtime monitoring. Since \SRAs subsume \RAs, nondeterministic \SRAs are strictly more expressive than deterministic ones, and language equivalence  is undecidable for nondeterministic \SRAs~\cite{Tzevelekos11}.

We now introduce the notions of \emph{simulation} and \emph{bisimulation} for \SRAs, which capture  whether one \SRA behaves ``at least as'' or ``exactly as'' another one.
\begin{definition}[(Bi)simulation for \SRAs]
A simulation $\ltssim$ on SRAs $\sra_1$ and $\sra_2$ is a binary relation $\ltssim$ on configurations such that $(p_1,v_1)\ltssim (p_2,v_2)$ implies:
\begin{itemize}[itemsep=1pt,topsep=3pt]
	\item if $p_1 \in F_1$ then $p_2 \in F_2$;
	\item for each transition $(p_1,v_1) \xrightarrow{a} (q_1,w_1)$ in $\clts(\sra_1)$, there exists a transition $(p_2,v_2) \xrightarrow{a} (q_2,w_2)$ in $\clts(\sra_2)$ such that $(q_1,w_1) \ltssim (q_2,w_2)$. 
\end{itemize}
A simulation $\ltssim$ is a \emph{bisimulation} if $\ltssim^{-1}$ is a also a simulation. We write $\sra_1 \prec \sra_2$ (resp.\ $\sra_1 \sim \sra_2$)  whenever there is a simulation (resp.\ bisimulation) $\ltssim$ such that $(q_{01},v_{01}) \ltssim (q_{02},v_{02})$, where $(q_{0i},v_{0i})$ is the initial configuration of $\sra_i$, for $i=1,2$.
\end{definition}
We say that an \SRA is \emph{complete} whenever for every configuration $(p,v)$ and  $a \in \els$ there is a transition $(p,v) \xrightarrow{a} (q,w)$ in $\clts(\sra)$. The following results connect similarity and language inclusion.
\begin{proposition}
If $\sra_1 \prec \sra_2$ then $\lang(\sra_1) \subseteq \lang(\sra_2)$. If $\sra_1$ and $\sra_2$ are deterministic and complete, then the other direction also holds.
\label{prop:det-sim-incl}
\end{proposition}

It is worth noting that given a deterministic \SRA we can define its \emph{completion} by adding transitions so that every value $a \in \els$ can be read from any state.

\begin{remark}
\RAs and \SFAs can be encoded as \SRAs on the same state-space:
\begin{itemize}[itemsep=1pt,topsep=3pt]
	\item 
	An \RA is encoded as an \SRA with all transition guards $\top$;
	\item 
 an \SFA can be encoded as an \SRA with $R = \emptyset$, with each SFA transition $p \xrightarrow{\varphi} q$ encoded as $p \tr{\varphi}{\emptyset,\emptyset,\emptyset} q$. Note that the absence of registers implies that the $\clts{}$ always has finitely many configurations.
\end{itemize}	
\SRAs are \emph{strictly more expressive} than both \RAs and \SFAs. For instance, the language $	\{ n_0 n_1 \dots n_k \mid n_0 = n_k, \mathsf{even}(n_i), n_i \in \mathbb{Z}, i = 1,\dots,k \}$ of finite sequences of even integers where the first and last one coincide, can be recognized by an \SRA, but not by an \RA or by an \SFA.
\label{rem:sfa-ra-sra}
\end{remark}

\medskip
\noindent\textbf{Boolean closure properties.} 
SRAs are closed under intersection and union. Intersection is given by a standard product construction whereas union is obtained by adding a new initial state that mimics the initial states of both automata.
\begin{proposition}[Closure under intersection and union]
Given  \SRAs $\sra_1$ and $\sra_2$, there are \SRAs $\sra_1 \cap \sra_2$ and $\sra_1 \cup \sra_2$ such that $\lang(\sra_1 \cap \sra_2) = \lang(\sra_1) \cap \lang(\sra_2)$ and 
$\lang(\sra_1 \cup \sra_2) = \lang(\sra_1) \cup \lang(\sra_2)$.
\label{prop:union-int}
\end{proposition}

\SRAs in general are not closed under complementation, because \RAs are not. However, we still have closure under complementation for a subclass of \SRAs. 
\begin{proposition}
Let $\sra$ be a complete and deterministic SRA, and let $\overline{\sra}$ be the SRA defined as $\sra$, except that its final states are $Q \setminus F$. Then $\lang(\overline{\sra}) = \els^{\star} \setminus \lang(\sra)$.
\end{proposition}

%% file: decid.tex
\section{Decidability properties}
\label{sec:decision}

In this section we will provide algorithms for checking determinism and emptiness for an \SRA,
and (bi)similarity of two \SRAs.
Our algorithms leverage \emph{symbolic} techniques that use the finite syntax of \SRAs to 
indirectly operate over the underlying configuration LTS, which can be infinite. 

\medskip
\noindent\textbf{Single-valued variant.}
To study decidability, it is convenient to restrict register assignments to \emph{injective} ones on non-empty registers, that is functions $v \colon R \to \els \cup \{\sharp\}$ such that $v(r) = v(s)$ and $v(r) \neq \sharp$ implies $r = s$. This is also the approach taken for \RAs in the seminal papers~\cite{Tzevelekos11,KaminskiF94}. Both for \RAs and \SRAs, this restriction does not affect expressivity.  We say that an \SRA is \emph{single-valued} if its initial assignment $v_0$ is injective on non-empty registers. For single-valued SRAs, we only allow two kinds of transitions:
\begin{description}[itemsep=0pt,topsep=1pt]
	\item[Read transition:] $p \tr{\varphi}{\checkreg{r}} q$ triggers when $a \in \den{\varphi}$ and $a$ is already stored in $r$. 
	\item[Fresh transition:] $p \tr{\varphi}{\fresh{r}} q$ triggers when the input $a \in \den{\varphi}$ and $a$ is \emph{fresh}, i.e., is not stored in any register. After the transition, $a$ is stored into $r$.
\end{description}
\SRAs and their single-valued variants have the same expressive power. Translating single-valued \SRAs to ordinary ones is straightforward:
\[
	p \tr{\varphi}{\checkreg{r}} q \; \Longrightarrow p \;\tr{\varphi}{\{r\},\emptyset,\emptyset} q 
	\qquad\quad
	p \tr{\varphi}{\fresh{r}} q \; \Longrightarrow p \;\tr{\varphi}{\emptyset,R,\{r\}} q 	
\]
The opposite translation requires a state-space blow up, because we need to encode register equalities in the states. 
\begin{theorem}
Given an \SRA $\sra$ with $n$ states and $r$ registers, there is a single-valued \SRA $\sra'$ with $\bigO(nr^r)$ states and $r+1$ registers such that $\sra \sim \sra'$. Moreover, the translation preservers determinism.
\label{thm:sra-to-sv}
\end{theorem}

\medskip
\noindent {\bf Normalization.}  
While our techniques are inspired by analogous ones for non-symbolic \RAs, \SRAs present an additional challenge: they can have arbitrary predicates on transitions. Hence, the values that each transition can read, and thus which configurations it can reach, depend on the history of past transitions and their predicates. 
This problem emerges when checking reachability and similarity, because a transition may be \emph{disabled} by particular register values, and so lead to unsound conclusions, a problem that does not exist in register automata.

\begin{example} 
Consider the \SRA below, defined over the BA of integers.

\centerline{    
$\sra = 
\begin{aligned}[c]
 \includegraphics[scale=.25]{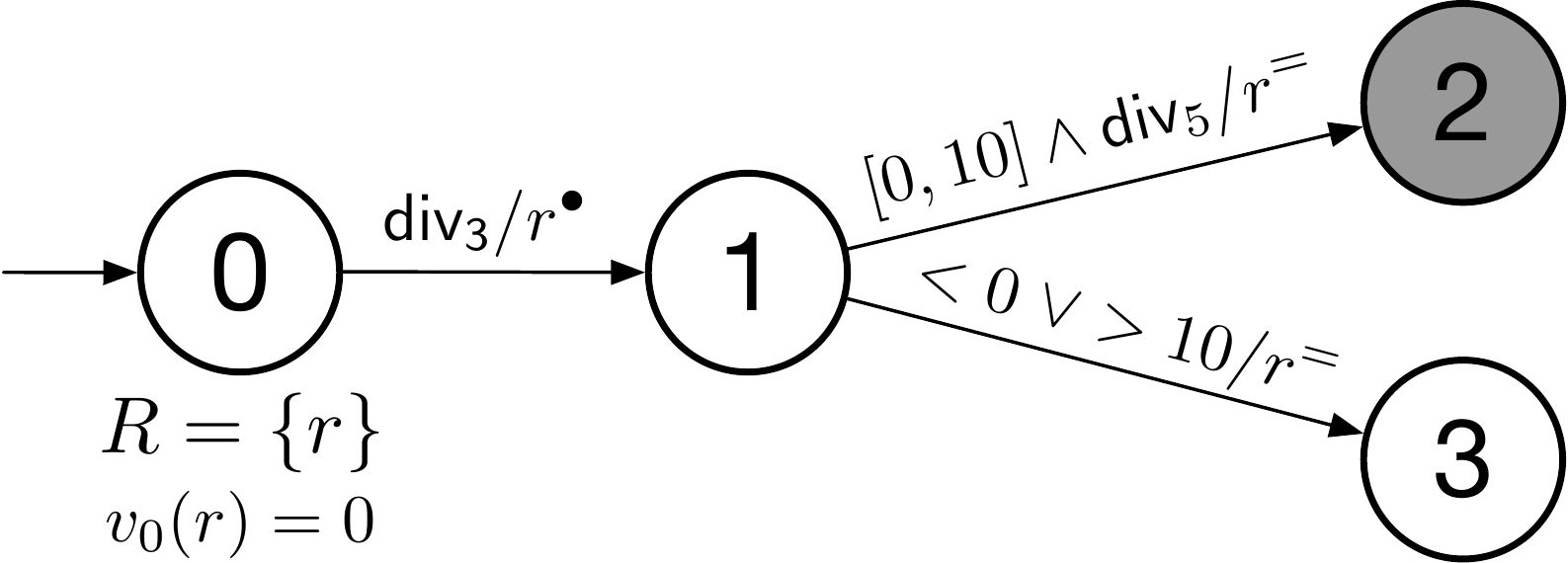}
\end{aligned}
$
}
\noindent
All predicates on transitions are satisfiable, yet $\lang(\sra) = \emptyset$. To go from 0 to 1, $\sra$ must read a value $n$ such that $\mathsf{div}_3(n)$ and $n \neq 0$ and then $n$ is stored into $r$. The transition from 1 to 2 can only happen if the content of $r$ also satisfies $\mathsf{div}_{5}(n)$ and $n \in [0,10]$. 
However, there is no $n$ satisfying $\mathsf{div}_3(n) \land n \neq 0 \land \mathsf{div_5}(n) \land n \in [0,10]$, hence the transition from 1 to 2 never happens.
\label{ex:sra-norm}
\end{example}
To handle the complexity caused by predicates, we introduce a way of \emph{normalizing} an \SRA to an equivalent one that 
\emph{stores additional information about input predicates}. 
We first introduce some notation and terminology.

A register abstraction $\theta$ for $\sra$, used to ``keep track'' of the domain of registers, is a family of predicates indexed by the registers $R$ of $\sra$. Given a register assignment $v$, we write $v \models \theta$ whenever $v(r) \in \den{\theta_r}$ for $v(r) \neq \sharp$, and $\theta_r = \bot$ otherwise.
Hereafter we shall only consider ``meaningful'' register abstractions, for which there is at least one assignment $v$ such that $v \models \theta$.

With the contextual information about register domains given by $\theta$, we say that a transition $p \tr{\varphi}{\ell} q \in \Delta$ is \emph{enabled by} $\theta$ whenever it has at least an instance $(p,v) \xrightarrow{a} (q,w)$ in $\clts(\sra)$, for all $v \models \theta$. Enabled transitions are important when reasoning about reachability and similarity.

Checking whether a transition has at least one realizable instance in the $\clts$ is difficult in practice, especially when $\ell = \fresh{r}$, because it amounts to checking whether $\den{\varphi} \setminus \img(v) \neq \emptyset$, for all injective $v \models \theta$. 

To make the check for enabledness practical we will use minterms. For a set of predicates $\Phi$, a \emph{minterm} is a minimal satisfiable Boolean combination of all predicates that occur in $\Phi$. Minterms are the analogue of atoms in a complete atomic Boolean algebra. 
E.g. the set of predicates $\Phi=\{x>2,x<5\}$ over the theory of linear
integer arithmetic has minterms
$\mint(\Phi)=\{x>2\wedge x<5,\ \neg x>2\wedge x<5,\
x>2\wedge \neg x<5\}$. Given $\psi \in \mint(\Phi)$ and $\varphi \in \Phi$, we will write $\varphi \inm \psi$ whenever $\varphi$ appears non-negated in $\psi$, for instance $(x > 2) \inm (x>2\wedge \neg x<5)$. A crucial property of minterms is that they do not overlap, i.e., $\sat(\psi_1 \land \psi_2)$ if and only if $\psi_1 = \psi_2$, for $\psi_1$ and $\psi_2$ minterms.

\begin{lemma}[Enabledness] 
Let $\theta$ be a register abstraction such that $\theta_r$ is a minterm, for all $r \in R$. If $\varphi$ is a minterm, then $p \tr{\varphi}{\ell} q$ is enabled by $\theta$ iff:

(1) if $\ell = \checkreg{r}$, then $\varphi = \theta_r$; \hspace{3ex} (2) 
 if $\ell = \fresh{r}$, then $\card{\den{\varphi}} > \regab(\theta,\varphi)$,

\noindent where $\regab(\theta,\varphi) = \card{\{ r \in R \mid \theta_r = \varphi \}}$ is the \# of registers with values from $\den{\varphi}$.
 \label{lem:minterm-enabled}
\end{lemma}
Intuitively, (1) says that if
 the transition reads a symbol stored in $r$ satisfying $\varphi$, the symbol must also satisfy $\theta_r$, the range of $r$. Because $\varphi$ and $\theta_r$ are minterms, this only happens when $\varphi = \theta_r$. (2) says that the enabling condition $\den{\varphi} \setminus \img(v) \neq \emptyset$, for all injective $v \models \theta$, holds if and only if there are fewer registers storing values from $\varphi$ than the cardinality of $\varphi$. That implies we can always find a fresh element in $\den{\varphi}$ to enable the transition. Registers holding values from $\varphi$ are exactly those $r \in R$ such that $\theta_r = \varphi$. 
Both conditions can be effectively checked: the first one is a simple predicate-equivalence check, while the second one amounts to checking whether $\varphi$ holds for at least a certain number $k$ of distinct elements. 
This can be achieved by checking satisfiability of $\varphi \land \neg \atom(a_1) \land \dots \land \neg \atom(a_{k-1})$, for $a_1,\dots,a_{k-1}$ distinct elements of $\den{\varphi}$.
\begin{remark}
Using single-valued \SRAs to check enabledness might seem like a restriction. However, if one would start from a generic \SRA, the process to check enabledness would contain an extra step: for each state $p$, we would have to keep track of all possible equations among registers. In fact, register equalities determine whether (i) register constraints of an outgoing transition are satisfiable; (ii) how many elements of the guard we need for the transition to happen, analogously to condition 2 of Lemma~\ref{lem:minterm-enabled}. Generating such equations is the key idea behind Theorem~\ref{thm:sra-to-sv}, and corresponds precisely to turning the \SRA into a single-valued one. 
\end{remark}
Given any \SRA, we can use the notion of register abstraction to build an equivalent \emph{normalized} SRA, where 
\rone states keep track of how the domains of registers change along transitions,
\rtwo transitions are obtained by breaking the one of the original \SRA into minterms and discarding the ones that are disabled according to Lemma~\ref{lem:minterm-enabled}.
In the following we write $\mint(\sra)$ for the minterms for the set of predicates $\{ \varphi \mid p \tr{\varphi}{\ell} q \in \Delta \} \cup \{ \atom(v_0(r)) \mid v_0(r) \in \els, r \in R \}$. 
Observe that an atomic predicate always has an equivalent minterm, hence we will use atomic predicates to define the initial register abstraction.
\begin{definition}[Normalized \SRA]
Given an \SRA $\sra$, its normalization $\nrm{\sra}$ is the \SRA $(R,\nrm{Q},\nrm{q_0},v_0,\nrm{F},\nrm{\Delta})$ where:
\begin{itemize}
	\item $\nrm{Q} = \{ \theta \mid \text{$\theta$ is a register abstraction over $\mint(\sra) \cup \{\bot\}$} \} \times Q$; we will write $\ns{q}{\theta}$ for $(\theta,q) \in \nrm{Q}$.
	\item $\nrm{q_0} = \ns{q_0}{\theta_0}$, where $(\theta_0)_r = \atom(v_0(r))$ if $v_0(r) \in \els$, and $(\theta_0)_r = \bot$ if $v_0(r) = \sharp$;
	\item $\nrm{F} = \{ \ns{p}{\theta} \in \nrm{Q} \mid p \in F \}$
	\item $
		\begin{aligned}[t]	
			\nrm{\Delta} = &\{ \ns{p}{\theta} \tr{\theta_r}{\checkreg{r}} \ns{q}{\theta} \mid p \tr{\varphi}{\checkreg{r}} q \in \Delta, \varphi \inm \theta_r  \} 
			\; \cup \\
			&\{ \ns{p}{\theta} \tr{\psi}{\fresh{r}} \ns{q}{\theta[r \mapsto \psi]} \mid  
			p \tr{\varphi}{\fresh{r}} q \in \Delta, \varphi \inm \psi, \card{\den{\psi}} > \regab(\theta,\psi) \}
		\end{aligned}	
		$
\end{itemize}
\label{def:nrm-cons}
\end{definition}
The automaton $\nrm{\sra}$ enjoys the desired property: each transition from $\ns{p}{\theta}$ is enabled by $\theta$, by construction.
$\nrm{\sra}$ is always \emph{finite}. In fact, suppose $\sra$ has $n$ states, $m$ transitions and $r$ registers. Then $\nrm{\sra}$ has at most $m$ predicates, and $\card{\mint(\sra)}$ is $\bigO(2^m)$. 
Since the possible register abstractions are $\bigO(r2^m)$, $\nrm{\sra}$ has $\bigO(n r 2^m)$ states and $\bigO(m r^2 2^{3m})$ transitions.

\begin{example}
We now show the normalized version of Example~\ref{ex:sra-norm}. The first step is computing the set $\mint(\sra)$ of minterms for $\sra$, i.e., the satisfiable Boolean combinations of $\{ \atom(0), \mathsf{div_3}, [0,10] \land \mathsf{div_5}, <0 \lor > 10\}$. For simplicity, we represent minterms as bitvectors where a 0 component means that the corresponding predicate is negated, e.g., $[1,1,1,0]$ stands for the minterm $\atom(0) \land ([0,10] \land \mathsf{div_3}) \land \mathsf{div_5} \land \neg( <0 \lor > 10)$. Minterms and the resulting \SRA $\nrm{\sra}$ are shown below. 
\centerline{
	\includegraphics[scale=.3]{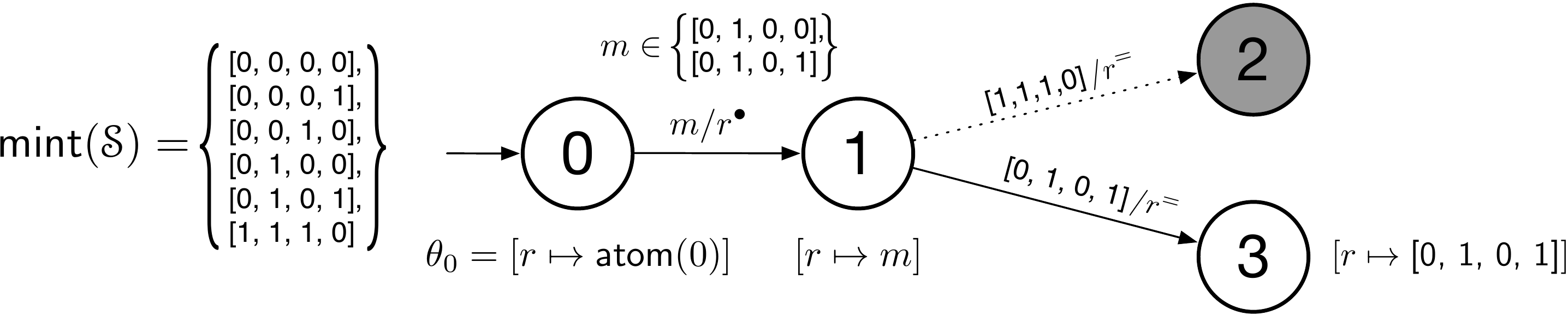}
}
On each transition we show how it is broken down to minterms, and for each state we show the register abstraction (note that state 1 becomes two states in $\nrm{\sra}$). The transition from 1 to 2 is \emph{not} part of $\nrm{\sra}$ -- this is why it is dotted. In fact, in every register abstraction $[r \mapsto m]$ reachable at state 1, the component for the transition guard $[0,10] \land \mathsf{div_5}$ in the minterm $m$ (3rd component) is 0, i.e., $([0,10] \land \mathsf{div_5}) \not \inm m$. Intuitively, this means that $r$ will never be assigned a value that satisfies $[0,10] \land \mathsf{div_5}$. As a consequence, the construction of Definition~\ref{def:nrm-cons} will not add a transition from 1 to 2. 
\end{example}
Finally, we show that the normalized \SRA behaves exactly as the original one.
\begin{proposition}
$(p,v) \sim (\ns{p}{\theta},v)$, for all $p \in Q$ and $v \models \theta$. Hence, $\sra \sim \nrm{\sra}$.
\label{prop:norm-bisim}
\end{proposition}

\subsubsection{Emptiness and Determinism.}
The transitions of $\nrm{\sra}$ are always enabled by construction, 
therefore every path in $\nrm{\sra}$ always corresponds to a run in $\clts(\nrm{\sra})$.
\begin{lemma}
The state $\ns{p}{\theta}$ is reachable in $\nrm{\sra}$ if and only if there is a reachable configuration $(\ns{p}{\theta},v)$ in $\clts(\nrm{\sra})$ such that $v \models \theta$. Moreover, if $(\ns{p}{\theta},v)$ is reachable, then all configurations $(\ns{p}{\theta},w)$ such that  $w \models \theta$ are reachable.
\label{lem:path-reach}
\end{lemma}
Therefore, using Proposition~\ref{prop:norm-bisim}, we can reduce the reachability and emptiness problems of $\sra$ to that of $\nrm{\sra}$.
\begin{theorem}[Emptiness]
There is an algorithm to decide reachability of any configuration of $\sra$, hence whether $\lang(\sra) = \emptyset$.
\end{theorem}
\begin{proof}
Let $(p,v)$ a configuration of $\sra$. To decide whether it is reachable in $\clts(\sra)$, we can perform a visit of $\nrm{\sra}$ from its initial state, stopping when a state $\ns{p}{\theta}$ such that $v \models \theta$ is reached. If we are just looking for a final state, we can stop at any state such that $p \in F$.
In fact, by Proposition~\ref{prop:norm-bisim}, there is a run in $\clts(\sra)$ ending in $(p,v)$ if and only if there is a run in $\clts(\nrm{\sra})$ ending in $(\ns{p}{\theta},v)$ such that $v \models \theta$. By Lemma~\ref{lem:path-reach}, the latter holds if and only if there is a path in $\nrm{\sra}$ ending in $\ns{p}{\theta}$. 
This algorithm has the complexity of a standard visit of $\nrm{\sra}$, namely $\bigO(n r 2^m + m r^2 2^{3m})$.
\qed
\end{proof}

Now that we characterized what transitions are reachable, we define what it means for a normalized \SRA to be deterministic and we show that determinism
is preserved by the translation from \SRA.
\begin{proposition}[Determinism]
$\nrm{\sra}$ is \emph{deterministic} if and only if for all reachable transitions $p \tr{\varphi_1}{\ell_1} q_1, p \tr{\varphi_2}{\ell_2} q_2 \in \nrm{\Delta}$ the following holds: $\varphi_1 \neq \varphi_2$ whenever either (1) $\ell_1 = \ell_2$ and $q_1 \neq q_2$, or; (2) $\ell_1 = \fresh{r}$, $\ell_2 = \fresh{s}$, and $r \neq s$;
\label{prop:red-det}
\end{proposition}
One can check determinism of an SRA by looking at its normalized version.
\begin{proposition}
$\sra$ is deterministic if and only if $\nrm{\sra}$ is deterministic.	
\label{prop:sra-nrm-det}
\end{proposition}

\subsubsection{Similarity and bisimilarity.}

We now introduce a symbolic technique to decide similarity and bisimilarity of \SRAs.
The basic idea is similar to \emph{symbolic (bi)simulation}~\cite{Tzevelekos11,MurawskiRT15} for RAs. Recall that \RAs are \SRAs whose transition guards are all $\top$.
Given two RAs $\sra_1$ and $\sra_2$ a symbolic simulation between them is defined over their state spaces $Q_1$ and $Q_2$, not on their configurations. For this to work, one needs to add an extra piece of information about how registers of the two states are related. More precisely, a symbolic simulation is a relation on triples $(p_1,p_2,\sigma)$, where $p_1 \in Q_1, p_2 \in Q_2$ and $\sigma \subseteq R_1 \times R_2$ is a \emph{partial injective} function. This function encodes constraints between registers: $(r,s) \in \sigma$ is an equality constraint between $r \in R_1$ and $s \in R_2$, and $(r,s) \notin \sigma$ is an inequality constraint. Intuitively, $(p_1,p_2,\sigma)$ says that all configurations $(p_1,v_1)$ and $(p_2,v_2)$ such that $v_1$ and $v_2$ satisfy $\sigma$ -- e.g., $v_1(r) = v_2(s)$ whenever $(r,s) \in \sigma$ -- are in the simulation relation $(p_1,v_1) \prec (p_2,v_2)$.
In the following we will use $v_1 \bowtie v_2$ to denote the function encoding constraints among $v_1$ and $v_2$, explicitly: $\sigma(r) = s$ if and only if $v_1(r) = v_2(s)$ and $v_1(r) \neq \sharp$. 
\begin{definition}[Symbolic (bi)similarity~\cite{Tzevelekos11}]
A symbolic simulation is a relation $\srasim \subseteq Q_1 \times Q_1 \times \Pow(R_1 \times R_2)$ such that if $(p_1,p_2,\sigma) \in \srasim$, then $p_1 \in F_1$ implies $p_2 \in F_2$, and if $p_1 \xrightarrow{\ell} q_1 \in \Delta_1$\footnote{We will keep the $\top$ guard implicit for succinctness.} then:
\begin{enumerate}
	\item if $\ell = \checkreg{r}$:
	\begin{enumerate} 
		\item if $r \in \dom(\sigma)$, then there is $p_2 \xrightarrow{\checkreg{\sigma(r)}} q_2 \in \Delta_2$ such that $(q_1,q_2,\sigma) \in \srasim$.
		\item if $r \notin \dom(\sigma)$ then there is $p_2 \xrightarrow{\fresh{s}} q_2 \in \Delta_2$ s.t. $(q_1,q_2, \sigma[r \mapsto s]) \in \srasim$.
	\end{enumerate}
	\item if $\ell=\fresh{r}$:
	\begin{enumerate} 
		\item for all $s \in R_2 \setminus \img(\sigma)$, there is $p_2 \xrightarrow{\checkreg{s}} q_2 \in \Delta_2$ such that $(q_1,q_2,\sigma[r \mapsto s]) \in \srasim$, and;
		\item there is $p_2 \xrightarrow{\fresh{s}} q_2 \in \Delta_2$ such that $(q_1,q_2,\sigma[r \mapsto s]) \in \srasim$.
	\end{enumerate}
\end{enumerate}
Here $\sigma[r \mapsto s]$ stands for $\sigma \setminus ( \sigma^{-1}(s), s ) \cup (r,s)$, which ensures that $\sigma$ stays injective when updated. 

Given a symbolic simulation $\srasim$, its inverse is defined as $\srasim^{-1} = \{ t^{-1} \mid t \in \srasim \}$, where $(p_1,p_2,\sigma)^{-1} = (p_2,p_1,\sigma^{-1})$. A \emph{symbolic bisimulation} $\srasim$ is a relation such that both $\srasim$ and $\srasim^{-1}$ are symbolic simulations. 
\label{def:sym-bisim}
\end{definition}
Case 1 deals with cases when $p_1$ can perform a transition that reads the register $r$. If $r \in \dom(\sigma)$, meaning that $r$ and $\sigma(r) \in R_2$ contain the same value, then $p_2$ must be able to read $\sigma(r)$ as well. If $r \notin \dom(\sigma)$, then the content of $r$ is fresh w.r.t.\ $p_2$, so $p_2$ must be able to read any fresh value --- in particular the content of $r$.
Case 2 deals with the cases when $p_1$ reads a fresh value. It ensures that $p_2$ is able to read all possible values that are fresh for $p_1$, be them already in some register $s$ -- i.e., $s \in R_2 \setminus \img(\sigma)$, case 2(a) -- or fresh for $p_2$ as well -- case 2(b). In all these cases, $\sigma$ must be updated to reflect the new equalities among registers. 

Keeping track of equalities among registers is enough for \RAs, because the actual content of registers does not determine the capability of a transition to fire (\RA transitions have implicit $\top$ guards). As seen in Example~\ref{ex:sra-norm}, this is no longer the case for \SRAs: a transition may or may not happen depending on the register assignment being compatible with the transition guard.

As in the case of reachability, normalized \SRAs provide the solution to this problem. We will reduce the problem of checking (bi)similarity of $\sra_1$ and $\sra_2$ to that of checking symbolic (bi)similarity on $\nrm{\sra_1}$ and $\nrm{\sra_2}$, with minor modifications to the definition. To do this, we need to assume that minterms for both $\nrm{\sra_1}$ and $\nrm{\sra_2}$ are computed over the union of predicates of $\sra_1$ and $\sra_2$.
\begin{definition}[\textsf{N}-simulation]
A \textsf{N}-simulation on $\sra_1$ and $\sra_2$ is a relation $\srasim \subseteq \nrm{Q_1} \times \nrm{Q_2} \times \Pow(R_1 \times R_2)$, defined as in Definition~\ref{def:sym-bisim}, with the following modifications: 
\begin{enumerate}[label=(\roman*),topsep=1pt,noitemsep]
	\item we require that $\ns{p_1}{\theta_1} \tr{\varphi_1}{\ell_1} \ns{q_1}{\theta_1'} \in \nrm{\Delta_1}$ must be matched by transitions $\ns{p_2}{\theta_2} \tr{\varphi_2}{\ell_2} \ns{q_2}{\theta_2'} \in \nrm{\Delta_2}$ such that $\varphi_2 = \varphi_1$.\label{rb:item1} \medskip
	\item we modify case 2 as follows (changes are underlined):
	\begin{enumerate}[label=2(\alph*)']
		\item for all $s \in R_2 \setminus \img(\sigma)$ \underline{such that $\varphi_1 = (\theta_{2})_s$}, there is $\ns{p_2}{\theta_2} \tr{\varphi_1}{\checkreg{s}} \ns{q_2}{\theta_2'} \in \nrm{\Delta_2}$ such that $(\ns{q_1}{\theta_1'},\ns{q_2}{\theta_2'},\sigma[r \mapsto s]) \in \srasim$, and;\label{rb:fresh-read}
		\item \underline{if $\regab(\theta_1,\varphi_1) + \regab(\theta_2,\varphi_1) < \card{\den{\varphi_1}}$, then} there is $\ns{p_2}{\theta_2} \tr{\varphi_1}{\fresh{s}} \ns{q_2}{\theta_2'} \in \nrm{\Delta_2}$ such that $(\ns{q_1}{\theta_1'},\ns{q_2}{\theta_2'},\sigma[r \mapsto s]) \in \srasim$.\label{rb:fresh-fresh}
	\end{enumerate}
	\label{rb:item2}
\end{enumerate}
A \textsf{N}-bisimulation $\srasim$ is a relation such that both $\srasim$ and $\srasim^{-1}$ are \textsf{N}-simulations. We write $\sra_1 \nrmsim \sra_2$ (resp.\ $\sra_1 \nrmbsim \sra_2$) if there is a \textsf{N}-simulation (resp.\ bisimulation) $\srasim$ such that $(\nrm{q_{01}},\nrm{q_{02}},v_{01} \bowtie v_{02}) \in \srasim$. 
\label{def:red-sym-bisim}
\end{definition} 
The intuition behind this definition is as follows. Recall that, in a normalized \SRA, transitions are defined over minterms, which cannot be further broken down, and are mutually disjoint. Therefore two transitions can read the same values if and only if they have the same minterm guard. 
Thus condition \ref{rb:item1} makes sure that matching transitions can read exactly the same set of values. Analogously, condition \ref{rb:item2} restricts how a fresh transition of $\nrm{\sra_1}$ must be matched by one of $\nrm{\sra_2}$: \ref{rb:fresh-read} only considers transitions of $\nrm{\sra_2}$ reading registers $s \in R_2$ such that $\varphi_1 = (\theta_{2})_s$ because, by definition of normalized \SRA, $\ns{p_2}{\theta_2}$ has no such transition if this condition is not met. Condition \ref{rb:fresh-fresh} amounts to requiring a fresh transition of $\nrm{\sra_2}$ that is enabled by both $\theta_1$ and $\theta_2$ (see Lemma~\ref{lem:minterm-enabled}), i.e., that can read a symbol that is fresh w.r.t.\ both $\nrm{\sra_1}$ and $\nrm{\sra_2}$.

\textsf{N}-simulation is sound and complete for standard simulation.
\begin{theorem}
$\sra_1 \prec \sra_2$ if and only if $\sra_1 \nrmsim \sra_2$.
\label{thm:nrmsim-sound-complete}
\end{theorem}
As a consequence, we can decide similarity of \SRAs via their normalized versions. \textsf{N}-simulation is a relation over a finite set, namely $\nrm{Q_1} \times \nrm{Q_2} \times \Pow(R_1 \times R_2)$, therefore \textsf{N}-similarity can always be decided in finite time. 
We can leverage this result to provide algorithms for checking language inclusion/equivalence for deterministic \SRAs (recall that they are undecidable for non-deterministic ones).

\begin{theorem}
Given two deterministic \SRAs $\sra_1$ and $\sra_2$, there are algorithms to decide $\lang(\sra_1) \subseteq \lang(\sra_2)$ and $\lang(\sra_1) = \lang(\sra_2)$.
\end{theorem}
\begin{proof}
By Proposition~\ref{prop:det-sim-incl} and Theorem~\ref{thm:nrmsim-sound-complete}, we can decide $\lang(\sra_1) \subseteq \lang(\sra_2)$ by checking $\sra_1 \nrmsim \sra_2$. This can be done algorithmically by iteratively building a relation $\srasim$ on triples that is an \textsf{N}-simulation on $\nrm{\sra_1}$ and $\nrm{\sra_2}$. The algorithm initializes $\srasim$ with $(\nrm{q_{01}},\nrm{q_{02}},$ $v_{01} \bowtie v_{02} )$, as this is required to be in $\srasim$ by Definition~\ref{def:red-sym-bisim}. Each iteration considers a candidate triple $t$ and checks the conditions for $\textsf{N}$-simulation. If satisfied, it adds $t$ to $\srasim$ and computes the next set of candidate triples, i.e., those which are required to belong to the simulation relation, and adds them to the list of triples still to be processed. If not, the algorithm returns $\lang(\sra_1) \not \subseteq \lang(\sra_2)$. The algorithm terminates returning $\lang(\sra_1) \subseteq \lang(\sra_2)$ when no triples are left to process. Determinism of $\sra_1$ and $\sra_2$, and hence of $\nrm{\sra_1}$ and $\nrm{\sra_2}$ (by Proposition~\ref{prop:sra-nrm-det}), ensures that computing candidate triples is deterministic. To decide $\lang(\sra_1) = \lang(\sra_2)$, at each iteration we need to check that both $t$ and $t^{-1}$ satisfy the conditions for $\textsf{N}$-simulation.

If $\sra_1$ and $\sra_2$ have, respectively, $n_1,n_2$ states, $m_1,m_2$ transitions, and $r_1,r_2$ registers, the normalized versions have $\bigO(n_1 r_1 2^{m_1})$ and $\bigO(n_2 r_2 2^{m_2})$ states. Each triple, taken from the finite set $\nrm{Q_1} \times \nrm{Q_2} \times \Pow(R_1 \times R_2)$, is processed exactly once, so the algorithm iterates $\bigO(n_1 n_2 r_1 r_2 2^{m_1 + m_2 + r_1 r_2})$ times.\qed
\end{proof}

%% file: appl.tex
\section{Evaluation}
\label{sec:eval}

We have implemented \SRAs in the open-source Java library SVPALib~\cite{svpalib}.
In our implementation, constructions are computed lazily when possible (e.g., the normalized \SRA for emptiness and (bi)similarity checks).
All experiments were performed on a machine with 3.5 GHz Intel Core i7 CPU with 16GB of RAM (JVM 8GB), with a timeout value of 300s. 
The goal of our evaluation is to answer the following research questions:
\begin{description}[topsep=1pt,noitemsep]
    \item [\textbf{Q1}:] Are \SRAs more succinct than existing models when processing strings over large but finite alphabets?  (\S~\ref{sec:eval:succinctness})
    \item [\textbf{Q2:}]  What is the performance of membership for deterministic \SRA and how 
    does it compare to the matching algorithm in \texttt{java.util.regex}? (\S~\ref{sec:eval:membership})
    \item [\textbf{Q3:}]  Are \SRAs decision procedures practical? (\S~\ref{sec:eval:decision})
\end{description}

\begin{figure}[t]%
\newcolumntype{C}[1]{>{\centering\arraybackslash}m{#1}}
\begin{tabular}{@{}C{.5\linewidth}@{\qquad} C{.5\linewidth}}
\multirow[t]{2}{*}[-6em]{	
\subfloat[][		Size of \SRAs vs \SFAs. (---) denotes the \SFA didn't fit in memory.
		$|$reg$|$ denotes how many different characters a register stored.
			\label{tab:srasfa}
]{
\setlength{\tabcolsep}{2pt}
	\centering
	\scriptsize
	\begin{tabular}{c|rrrr|rr@{\quad}}
						 & \multicolumn{4}{c|}{\textbf{\SRA}} & \multicolumn{2}{c}{\textbf{SFA}}  \\
						 & \multicolumn{1}{c}{states} & \multicolumn{1}{c}{tr} & \multicolumn{1}{c}{reg}  & 
						 	\multicolumn{1}{c|}{$|$reg$|$} & \multicolumn{1}{c}{states} & \multicolumn{1}{c}{tr} \\		
		\hline
	 IP2 &  44 & 46 & 3 & 10 & 4,013 & 4,312\\ 
	 IP3 &  44 & 46 & 4 & 10 & 39,113 & 42,112\\
	 IP4 &  44 & 46 & 5 & 10 & 372,113 & 402,112\\
	 IP6 &  44 & 46 & 7 & 10 & --- & ---\\
	 IP9 &  44 & 46 & 10 & 10 & --- & ---\\
	 Name-F & 7 & 10 & 2 & 26 & 201 & 300\\
	 Name-L & 7 & 10 & 2 & 26 & 129 & 180\\	 
	 Name & 7 & 10 & 3 & 26 & 3,201 & 4,500\\
     XML & 12 & 16 & 4 & 52 & --- & ---\\
     Pr-C2  & 26 & 28 & 3 & $2^{16}$ &  --- & ---\\
     Pr-C3 & 28 & 30 & 4 & $2^{16}$ &--- & ---\\
     Pr-C4    & 30 & 32 & 5 & $2^{16}$ &--- & ---\\  
     Pr-C6    & 34 & 36 & 7 & $2^{16}$ &--- & ---\\  
     Pr-C9   & 40 & 42 & 10 & $2^{16}$ &--- & ---\\ 
     Pr-CL2 & 26 & 28 & 3 & $2^{16}$ &--- & ---\\
     Pr-CL3 & 28 & 30 & 4 & $2^{16}$ &--- & ---\\
     Pr-CL4  & 30 & 32 & 5 & $2^{16}$ &--- & ---\\    
     Pr-CL6    & 34 & 36 & 7 & $2^{16}$ &--- & ---\\  
     Pr-CL9    & 40 & 42 & 10 & $2^{16}$ &--- & --- 
	\end{tabular}
}
}
&
\subfloat[][Performance of decision procedures. In the table $\lang_i = \lang(\sra_i)$, for $i=1,2$.\label{tab:decision}]{
	\centering
	\scriptsize
	\begin{tabular}{cc|rrr}
		\SRA $\sra_1$ & \SRA $\sra_2$ & \multicolumn{1}{c}{$\lang_1{=}\emptyset$}  & \multicolumn{1}{c}{$\lang_1{=}\lang_1$} & \multicolumn{1}{c}{$\lang_2 \subseteq \lang_1$}  \\		
		\hline
		Pr-C2 & Pr-CL2 & 0.125s & 0.905s & 3.426s \\
		Pr-C3 & Pr-CL3 & 1.294s & 5.558s & 24.688s \\
		Pr-C4 & Pr-CL4 & 13.577s & 55.595s & --- \\
		Pr-C6 & Pr-CL6 & --- & --- & --- \\
		Pr-CL2 & Pr-C2 & 1.067s & 0.952s & 0.889s \\
		Pr-CL3 & Pr-C3 & 10.998s & 11.104s & 11.811s \\
		Pr-CL4 & Pr-C4 & --- & --- & --- \\
		Pr-CL6 & Pr-C6 & --- & --- & --- \\
		IP-2 & IP-3 & 0.125s & 0.408s & 1.845s \\	
		IP-3 & IP-4 & 1.288s & 2.953s & 21.627s \\	
		IP-4 & IP-6 & 18.440s & 42.727s & --- \\	
		IP-6 & IP-9 & --- & --- & --- \\	
	\end{tabular}
}\\
&\subfloat[][\SRA membership and Java \texttt{regex} matching performance. Missing data points for Java are stack overflows.]{
\includegraphics{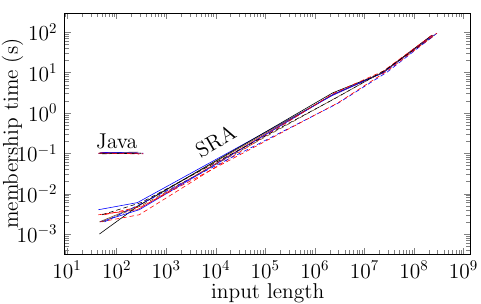}
\label{fig:sra-membership}
}
\end{tabular}
\caption{Experimental results.}
\end{figure}

\paragraph{\bf Benchmarks.} We focus on regular expressions with back-references,
therefore all our benchmarks operate over the Boolean algebra of Unicode characters with interval---i.e., the set of
characters is the set of all $2^{16}$ UTF-16 characters and the predicates are union of intervals (e.g., \texttt{[a-zA-Z]}).\footnote{
Our experiments are over finite alphabets, but the Boolean algebra can be infinite by taking the alphabet to be positive integers and
allowing intervals to contain $\infty$ as upper bound. 
This modification does not affect the running time of our procedures, therefore we do not report it.}
Our benchmark set contains 19 \SRAs that represent variants of regular expressions with back-references obtained
from the regular-expression crowd-sourcing website RegExLib~\cite{regexlib}. 
The expressions check whether inputs have, for example, matching first/last name initials or both (Name-F, Name-L and Name), correct Product Codes/Lot number of total length $n$ (Pr-C$n$, Pr-CL$n$), matching XML tags (XML), and IP addresses that match for $n$ positions (IP$n$).
We also create variants of the product benchmark presented in Section~\ref{sec:mot-ex} where we vary the numbers of characters in the code and lot number.
All the \SRAs are deterministic.

\subsection{Succinctness of \SRAs vs \SFAs}
\label{sec:eval:succinctness}
In this experiment, we relate the size of \SRAs over finite alphabets to the size
 of the smallest equivalent \SFAs.
For each \SRA, we construct the equivalent \SFA by equipping the state space with the values stored in the registers at each step (this
construction effectively builds the configuration LTS).
Figure~\ref{tab:srasfa} shows the results.
As expected, \SFAs tend to blow up in size when the \SRA contains multiple registers or complex register values.
In cases where the register values range over small sets (e.g., \texttt{[0{-}9]}) it is often feasible to build an \SFA equivalent to the \SRA,
but the construction always yields very large automata. In cases where the registers can assume many values (e.g., $2^{16}$) \SFAs become 
prohibitively large and do not fit in memory.
To answer \textbf{Q1}, even for finite alphabets,
\textbf{it is not feasible to compile \SRAs to \SFAs}. Hence, \SRAs are a succinct model.

\subsection{Performance of membership checking}
\label{sec:eval:membership}
In this experiment, we measure the performance of \SRA membership, and we compare it with the performance of the \texttt{java.util.regex} matching algorithm.
For each benchmark, we generate inputs of length varying between approximately 100 and 
10$^8 $ characters and 
measure the time taken to check membership.
Figure~\ref{fig:sra-membership} shows the results. The performance of SRA (resp. Java) is not particularly affected by the size of the expression.
Hence, the lines for different expressions mostly overlap. 
As expected, for \SRAs the time taken to check membership grows linearly in the size of the input (axes are log scale).
Remarkably, even though our implementation does not employ particular input processing optimizations,
it can still check membership for strings with tens of millions of characters in less than 10 seconds. We have found that our implementation is more efficient than the Java \texttt{regex} library, matching the same input an average of 50 times faster than \texttt{java.util.regex.Matcher}. \texttt{java.util.regex.Matcher} seems to make use of a recursive algorithm to match back-references, which means it does not scale well. Even when given the maximum stack size, the JVM will return a Stack Overflow for inputs as small as 20,000 characters. Our implementation can match such strings in less than 2 seconds.
To answer \textbf{Q2},
\textbf{deterministic \SRAs can be  efficiently executed on large inputs and perform better than the \texttt{java.util.regex} matching algorithm}.

\vspace{-2ex}
\subsection{Performance of decision procedures}
\label{sec:eval:decision} 
In this experiment, we measure the performance of \SRAs simulation and bisimulation algorithms.
Since all our \SRAs are deterministic, these two checks correspond to language equivalence and inclusion.
We select pairs of benchmarks for which the above tests are meaningful (e.g., variants of the problem discussed at the end of Sec.~\ref{sec:mot-ex}).
The results are shown in Figure~\ref{tab:decision}.
As expected, due to the translation to single-valued \SRAs, our decision procedures do not scale well in the number of registers.
This is already the case for classic register automata and it is not a surprising result.
However, our technique can still check equivalence and  inclusion for regular expressions that no existing tool can handle.
To answer \textbf{Q3},
\textbf{bisimulation and simulation algorithms for \SRAs only scale to small numbers of registers}.

%% file: related.tex
\section{Conclusions}
\label{sec:related}

In this paper we have presented \emph{Symbolic Register Automata}, a novel class of automata that can handle complex alphabet theories while allowing symbol comparisons for equality. \SRAs encompass -- and are strictly more powerful -- than both Register and Symbolic Automata. We have shown that they enjoy the same closure and decidability properties of the former, despite the presence of arbitrary guards on transitions, which are not allowed by \RAs. Via a comprehensive set of experiments, we have concluded that \SRAs are vastly more succinct than \SFAs and membership is efficient on large inputs. Decision procedures do not scale well in the number of registers, which is already the case for basic \RAs.

\noindent
\textbf{Related work.}
\RAs were first introduced in \cite{KaminskiF94}. There is an extensive literature on register automata, their formal languages and decidability properties~\cite{DemriL09,NevenSV04,SakamotoI00,CasselHJMS15,MurawskiRT18}, 
including variants with \emph{global freshness}~\cite{Tzevelekos11,MurawskiRT15} and totally ordered data~\cite{BenediktLP10,FigueiraHL16}. \SRAs are based on the original model of~\cite{KaminskiF94}, but are much more expressive, due to the presence of guards from an arbitrary decidable theory.

In recent work, variants over richer theories have appeared. In \cite{ChenLTW17} \RA over rationals were introduced. They allow for a restricted form of linear arithmetic among registers (\RAs with arbitrary linear arithmetic subsume two-counter automata, hence are undecidable). \SRAs do not allow for operations on registers, but encompass a wider range of theories without any loss in decidability. Moreover, \cite{ChenLTW17} does not study Boolean closure properties.
In~\cite{CasselHJS16,IsbernerHS14}, \RAs allowing guards over a range of theories -- including (in)equality, total orders and increments/sums -- are studied. Their focus is different than ours as they are  interested primarily in \emph{active learning} techniques, and several restrictions are placed on models for the purpose of the learning process. We can also relate \SRAs with \emph{Quantified Event Automata}~\cite{BarringerFHRR12}, which allow for guards and assignments to registers on transitions. However, in QEA guards can be arbitrary, which could lead to several problems, e.g. undecidable equivalence.

Symbolic automata were first introduced in~\cite{5477051} and many variants of them 
have been proposed~\cite{DAntoniV17}.  The one that is closer to \SRAs is Symbolic Extended Finite Automata (\SEFA)~\cite{DAntoni2015}. \SEFAs are \SFAs in which transition can read more than one character at a time. A transition
of arity $k$ reads $k$ symbols which are consumed if they satisfy the predicate $\varphi(x_1,\ldots,x_k)$.  \SEFAs allow arbitrary $k$-ary predicates over the input theory, which results in most problems being undecidable (e.g., equivalence and intersection emptiness) and in the model not being closed under Boolean operations.
Even when deterministic, \SEFAs are not closed under union and intersection.
In terms of expressiveness, \SRAs and \SEFAs are incomparable. \SRAs can only use equality, but can compare symbols at arbitrary points in the input while \SEFAs can only compare symbols within a constant window, but using arbitrary predicates.

Several works study matching techniques for extended regular expressions~\cite{BecchiC08,BispoSCV06,ReidenbachS10,Komendantsky12}. These works introduce automata models with ad-hoc features for extended regular constructs -- including back-references -- but focus on efficient matching, without studying closure and decidability properties. It is also worth noting that \SRAs are not limited to alphanumeric or finite alphabets. On the negative side, \SRAs cannot express capturing groups of an unbounded length, due to the finitely many registers. This limitation is essential for decidability.

\noindent 
\textbf{Future work.}
In \cite{MurawskiRT18} a polynomial algorithm for checking language equivalence of deterministic \RAs is presented. This crucially relies on closure properties of symbolic bisimilarity, some of which are lost for \SRAs. We plan to investigate whether this algorithm can be adapted to our setting. 
Extending \SRAs with more complex comparison operators other than equality (e.g., a total order $<$) is an interesting research question, but most extensions of the model quickly lead to undecidability.
We also plan to study active automata learning for \SRAs, building on techniques for \SFAs~\cite{ArgyrosD18}, \RAs~\cite{BolligHLM13,CasselHJS16,IsbernerHS14} and nominal automata~\cite{MoermanS0KS17}.

%% file: appendix.tex
\begin{appendix}

\section{Proofs}

\begin{proof}[of Proposition~\ref{prop:union-int}]
Intersection and union are defined as follows:
\begin{enumerate} \item 
$\sra_1 \cap \sra_2 = (R_\cap, Q_\cap, q_{0\cap},v_{0\cap},F_\cap,\Delta_\cap)$ where $R_\cap = R_1 \uplus R_2$; $Q_\cap = Q_1 \times Q_2$; $q_{0\cap} = (q_{01},q_{02})$; $v_{0\cap} = [v_{01}, v_{02}]$; $F_\cap = F_1 \times F_2$; and 
\[
	\Delta_\cap = \{ (p_1,p_2) \tr{\varphi_1 \land \varphi_2}{\ell_1 \cup \ell_2} (q_1,q_2) \mid \forall i=1,2 \colon p_i \tr{\varphi_i}{\ell_i} q_i \in \Delta_i \}
\]
where $(E_1,I_1,U_1) \cup (E_2,I_2,U_2) = E_1 \cup E_2, I_1 \cup I_2,U_1 \cup U_2$.
\item $\sra_1 \cup \sra_2 = (R_\cup, Q_\cup, q_{0\cup},v_{0\cup},F_\cup,\Delta_\cup)$ where $R_\cup = R_1 \uplus R_2$; $Q_\cup = Q_1 \cup Q_2 \cup \{\hat{q}_0\}$; $q_{0\cup} = \hat{q}_0$; $v_{0\cup} = [v_{01}, v_{02}]$; $F_\cup = F_1 \cup F_2 \cup I$, where $I = \{ \hat{q}_0 \}$ if $q_{0i} \in F_i$, for some $i \in \{1,2\}$, and $I = \emptyset$ otherwise. Transitions are given by
\[
	\Delta_\cup = \bigcup_{i \in \{1,2\}} \Delta_i \cup 
	\{ \hat{q}_0 \tr{\varphi}{E,I,U} q \mid q_{0i} \tr{\varphi}{E,I,U} q \in \Delta_i\} 
\]
\end{enumerate}
\end{proof}

\begin{proof}[of Theorem~\ref{thm:sra-to-sv}]
\newcommand{\rr}{(\textsc{reg})\xspace}
\newcommand{\fr}{(\textsc{fresh})\xspace}
\newcommand{\nr}{(\textsc{nop})\xspace}	
%
For notational convenience, we introduce an additional label for single-valued \SRAs: $p \tr{\varphi}{\bullet} q$ whenever a fresh symbol is read, but not stored anywhere. This can be translated back to ordinary single-valued transitions by adding a new register $\hat{r}$, replacing $p \tr{\varphi}{\bullet} q$ with a pair of transitions \begin{tikzcd}[column sep=8ex] p \arrow[r,"{\varphi/\fresh{\hat{r}}}",bend left=20] \arrow[r,"{\varphi/\checkreg{\hat{r}}}"',bend right=20] & q\end{tikzcd}, and by adding $p \tr{\varphi}{\checkreg{\hat{r}}} q$ for every transition $p \tr{\varphi}{\fresh{r}} q$.

Let $S = (R,Q,q_0,v_0,F,\Delta)$. The SRA $\sra' = (R,Q',q_0',v_0',F',\Delta')$ is defined as follows:
\begin{itemize}
	\item $Q' = \{(q,f) \mid q \in Q, f \colon R \to R \}$;
	\item 
	$v_0'$ is any function $R \to \els \cup \{\sharp\}$ injective on non-empty registers such that $\img(v_0') = \img(v_0)$;
	\item $q_0' = (q_0,f_0)$, where $f_0$ is such that $v_0 = v_0' \circ f_0$;
	\item $F' = \{ (q,f) \in Q' \mid q \in F\}$;
	\item $\Delta'$ is generated by the following rules
\begin{gather*}
	\rr \; \frac{p \tr{\varphi}{E,I,U} q \quad S \propto E,I \quad f^{-1}(r) = S}
	{(p,f) \tr{\varphi}{\checkreg{r}} (q,f[U \mapsto r])} %
	\qquad	
	\nr \; \frac{p \tr{\varphi}{E,I,\emptyset} q \quad \emptyset \propto E,I }
	{(p,f) \tr{\varphi}{\bullet} (q,f)} 	
	\\[1ex]
	\fr \; \frac{p \tr{\varphi}{E,I,U} q \quad \emptyset \propto E,I \quad f^{-1}(r) \subseteq U \quad U \neq \emptyset}
	{(p,f) \tr{\varphi}{\fresh{r}} (q,f[U \mapsto r])} 
\end{gather*}
where $S \propto E,I$ whenever $E \subseteq S$ and $I \cap S = \emptyset$.
\end{itemize}
Moreover, for \fr, we assume that the choice of $r$ only depends on $f$ and $U$, i.e., given any two transitions $p_1 \tr{\varphi_1}{E_1,I_1,U} q_1$ and $p_2 \tr{\varphi_2}{E_2,I_2,U} q_2$ of $\sra$, the rule \fr produces a unique pair $(p_1,f) \tr{\varphi_1}{\fresh{r}} (q_1,f[U \mapsto r])$ and $(p_2,f) \tr{\varphi_2}{\fresh{r}} (q_2,f[U \mapsto r])$. This can be achieved by ordering the registers and picking the least $r$ such that $f^{-1}(r) \subseteq U$.


We now prove that the following relation is a bisimulation between $\sra$ and $\sra'$:
\[
	\srasim = \{ ((p,v),((p,f),w) \mid v = w \circ f \}
\]
We first prove the following, which will be useful in all the cases below. 

\begin{lemma}
Let $X \subseteq R$ and let $v = w \circ f$, with $v$, $w$, $f$ as above. For any $r \in R$ such that $f^{-1}(r) \subseteq X$, and any $a \in \els$, we have
\[
	v[X \mapsto a] = w[r \mapsto a] \circ f[X \mapsto r]
\]
\label{lem:decomp}
\end{lemma}
\begin{proof}
We proceed by cases. Let $s \in X$. Then we have $v[X \mapsto a](s) = a$ and $(w[r \mapsto a] \circ f[X \mapsto r])(s) = w[r \mapsto a](r) = a$. If $s \notin X$, then we must have $f(s) \neq r$, because $f^{-1}(r) \subseteq X$, so $(w[r \mapsto a] \circ f[X \mapsto r])(s) = w[r \mapsto a](f(s)) = w(f(s)) = v(s) = v[X \mapsto a](s)$, by the hypothesis $v = w \circ f$.
\end{proof}

Clearly we have $((q_0,v_0), ((p,f_0),v_0')) \in \srasim$, by construction. We first prove that $\srasim$ is a simulation. Suppose we have $(p,v) \xrightarrow{a} (q,v') \in \clts(\sra)$, instance of $p \tr{\varphi}{E,I,U} q$, so $v^{-1}(a) \propto E,I$ and $v' = v[U \mapsto a]$. We need to find a matching transition in $\clts(\sra')$ from any $((p,f),w)$ such that $w$ is injective on non-empty registers and $v = w \circ f$. We have two cases:
\begin{itemize}
	\item Suppose $a \in \img(v)$, and let $S = v^{-1}(a)$. Since $v = w \circ f$, there must be $r \in R$ (unique, by injectivity of $w$) such that $f(S) = \{r\}$. Therefore, recalling the assumption $S \propto E,I$, by \rr there is $p \tr{\varphi}{\checkreg{r}} q$ and, since $w(r) = (w \circ f)(E) = v(E) = a$, there is $((p,f),w) \xrightarrow{a} ((q,f[U \mapsto r]),w) \in \clts(\sra')$. We have 
	\begin{align*}
		v[U \mapsto a] &= v[U \cup S \mapsto a] && (v(S) = a) \\
		&= (w \circ f)[E \cup S \mapsto a] && (v = w \circ f ) \\
		&= w[r \mapsto a] \circ f[S \cup U \mapsto r] && (\text{Lemma~\ref{lem:decomp} and $f^{-1}(r) = S \subseteq S \cup U$}) \\
		&= w \circ f[U \mapsto r] && (a = w(r), f^{-1}(r) = S) 
	\end{align*}
	from which we get $(q,v[U \mapsto a]) \srasim ((q,f[U \mapsto r]),w)$.
	
	\item Suppose $a \in \den{\varphi} \setminus \img(v)$. Then we have $\emptyset \propto E,I$. If $a \in \img(w)$, with $r = w(a)$, we must have $f^{-1}(r) = \emptyset$, therefore $f^{-1}(r) \propto E,I$. We can now use \rr to obtain $(p,f) \tr{\varphi}{\checkreg{r}} (q,f[U \mapsto r]) \in \Delta'$. This transition has an instance $((p,f),w) \xrightarrow{a} ((q,f[U \mapsto r]),w)$ in $\clts(\sra')$. We have $v[U \mapsto a] = w \circ f[U \mapsto r]$, because $w(r) = a$, so $(q,v[U \mapsto a]) \srasim ((q,f[U \mapsto r]),w)$.	
	If $a \notin \img(w)$, we can apply either $\rr$ or $\fr$, depending on $U$:
			\begin{itemize} 
				\item Suppose $U \neq \emptyset$, and consider $r \in R$ such that that $f^{-1}(r) \subseteq U$. Such $r$ must exist: if $f$ is injective then it is obvious (this crucially depends on $U \neq \emptyset$), otherwise there is $r$ such that $f^{-1}(r) = \emptyset \subseteq U$. Therefore by \fr there is a transition $(p,f) \tr{\varphi}{\fresh{r}} (q,f[U \mapsto r]) \in \Delta'$, which has an instance $((p,f),w) \xrightarrow{a} ((q,f[U \mapsto r]),w[r \mapsto a])$ in $\clts(\sra')$. By Lemma~\ref{lem:decomp} and $f^{-1}(r) \subseteq U$ (see previous point), we have $v[U \mapsto a] = w[r \mapsto a] \circ f[U \mapsto r]$, from which $(q,v[U \mapsto a]) \srasim ((q,f[U \mapsto r]),w[r \mapsto a])$ follows.
				\item Suppose $U = \emptyset$. Then by \nr there is $(p,f) \tr{\varphi}{\bullet} (q,f) \in \Delta'$. This transition has an instance $((p,f),w) \xrightarrow{a} ((q,f),w)$ in $\clts(\sra')$. By assumption $v = w \circ f$, so $(q,v) \srasim ((q,f),w)$.
			\end{itemize}
			
	\end{itemize}
	Now we prove that $\srasim^{-1}$ is a simulation. Suppose we have $((p,f),w) \xrightarrow{a} ((q,f'),w')$ in $\clts(\sra')$. We need to find a matching transition in $\clts(\sra)$ from any $(p,v)$ such that $v = w \circ f$. We proceed by cases:
	\begin{itemize}
		\item Suppose $a \in \img(w)$, with $w(r) = a$. Then there must be $(p,f) \tr{\varphi}{\checkreg{r}} (q,f')$ in $\Delta'$ and we have $w' = w$. This transition is given by \rr, so  we have $p \tr{\varphi}{E,I,U} q \in \Delta$,
			$f' = f[U \mapsto r]$ and $f^{-1}(r) = S \propto E,I$. By injectivity of $w$ and $v = w \circ f$, we must have $v^{-1}(a) = S$. Therefore there is $(p,v) \xrightarrow{a} (q,v[U \mapsto a])$ in $\clts(\sra)$. We have $v[U \mapsto a] = w \circ f[U \mapsto r]$, so $((q,f[U \mapsto r]),w) \srasim^{-1} (q,v[U \mapsto a])$.
		\item Suppose $a \in \den{\varphi} \setminus \img(w)$. Then we must have $a \notin \img(v)$, because $\img(v) \subseteq \img(w)$. We have two cases:
		\begin{itemize}
			\item Suppose $(p,f) \tr{\varphi}{\fresh{r}} (q,f') \in \Delta'$. Then we have $w' = w[r \mapsto a]$. By \fr, there is $p \tr{\varphi}{E,I,U} q \in \Delta$. Because $a \notin \img(v)$ implies $v^{-1}(a) = \emptyset$, and we have $\emptyset \propto E,I$, there is $(p,v) \xrightarrow{a} (q,v[U \mapsto a])$ in $\clts(\sra)$. By $f^{-1}(r) \subseteq U$ and Lemma~\ref{lem:decomp}, there is $v[U \mapsto a] = w[r \mapsto a] \circ f[U \mapsto r]$, so $((q,f[U \mapsto r]),w[r \mapsto a]) \srasim^{-1} (q,v[U \mapsto a])$. 
			\item Suppose $(p,f) \tr{\varphi}{\bullet} (q,f) \in \Delta'$. Then by \nr there must be $p \tr{\varphi}{E,I,\emptyset} q \in \Delta$. Because $a \notin \img(v)$, we have $v^{-1}(a) = \emptyset$, and because $\emptyset \propto E,I$ there is $(p,v) \xrightarrow{a} (q,v)$ in $\clts(\sra)$. By the assumption $v = w \circ f$, we get $((q,f),v) \srasim^{-1} (q,w)$.
		\end{itemize}
	\end{itemize}
Now we show that, if $\sra$ is deterministic, so is $\sra'$. We proceed by contradiction. Suppose $\sra'$ is not deterministic, i.e., there is a reachable configuration $((p,f),w)$ in $\clts(\sra')$ and two transitions $((p,f),w) \xrightarrow{a} ((q_1,f_1),w_1)$ and $((p,f),w) \xrightarrow{a} ((q_2,f_2),w_2)$ such that $((q_1,f_1),w_1) \neq ((q_2,f_2),w_2)$. Then these transitions must be instances of $(p,f) \tr{\varphi_i}{\ell_i} (q_i,f_i) \in \Delta'$, for $i=1,2$. We have three cases:
\begin{itemize}
	\item if $\ell_1 = \checkreg{r}$, then we must also have $\ell_2 = \checkreg{r}$ (and vice versa), because $a = w(r)$, and $w_1 = w_2 = w$. Since $\sra \sim \sra'$, there must be a reachable configuration $(p,v)$ in $\clts(\sra)$ such that $v = w \circ f$, and two transitions $(p,v) \xrightarrow{a} (q_i,v_i)$ such that $v_i = w \circ f_i$, for $i=1,2$. Now, we must have either $q_1 \neq q_2$ or $v_1 \neq v_2$. In fact, suppose $q_1 = q_2$ and $v_1 = v_2$, then we would have
	\begin{align*}
		v_1 = w \circ f_1 &\implies w \circ f_1 = w \circ f_2  && (v_2 = w \circ f_2 \land v_1 = v_2 )\\
		&\implies f_1 = f_2 && (\text{injectivity of $w$})
	\end{align*}
	which implies $((q_1,f_1),w_1) = ((q_2,f_2),w_2)$, contradicting our initial assumption on non-determinism of $\sra$. Therefore $(q_1,v_1) \neq (q_2,v_2)$, and $\sra$ is non-deterministic as well.
	\item if $\ell_1 = \fresh{r}$, then we have $\ell_2 \in \{ \bullet, \fresh{s}\}$ and $w_1 = w[r \mapsto a]$, $w_2 \in \{ w_2, w[s \mapsto a] \}$, because $a \notin \img(w)$. Suppose $\ell_2 = \bullet$. By the definition of $\Delta'$, the two transitions of $\sra'$ we are considering must correspond to transitions $p \tr{\varphi_i}{E_i,I_i,U_i} q_i \in \Delta$ such that $\emptyset \propto E_i,I_i$ where $U_2$ can be empty, and $f_i = f[U_i \mapsto r]$, for $i=1,2$. Let $(p,v)$ be a reachable configuration in $\clts(\sra)$ such that $v = w \circ f$, which must exist by $\sra \sim \sra'$. Then those transitions in $\Delta$ have instances $(p,v) \xrightarrow{a} (q_i,v[U_i \mapsto a])$, for $i=1,2$. We distinguish the two cases:
	\begin{itemize}
		\item If $\ell_2 = \bullet$, then $v[U_2 \mapsto a] = v \neq v[U_1 \mapsto a]$, because $a \notin \img(w)$ and $v = w \circ f$ imply $a \notin \img(v)$, and because $U_1 \neq \emptyset$, by the premise of \fr. Therefore we have $(q_1,v[U_1 \mapsto a]) \neq (q_2,v)$, from which non-determinism of $\sra$ follows.
		\item If $\ell_2 = \fresh{s}$, then we must have either $q_1 \neq q_2$ or $v[U_1 \mapsto a] \neq v[U_2 \mapsto a]$. In fact, suppose $q_1 = q_2$ and $v[U_1 \mapsto a] = v[U_2 \mapsto a]$. Because $a \notin \img(v)$, the latter equality implies $U_1 = U_2$, thus $r = s$, because we assumed that \fr picks a unique $r$ for each pair $(f,U_1)$ ($ = (f,U_2)$). It follows that $f_1 = f_2$ and $w_1 = w_2$, from which we get the contradiction $((q_1,f_1),w_1) = ((q_2,f_2),w_2)$. 
	\end{itemize}
	
	\item the case $\ell_1 = \bullet$ is symmetrical to the one above.
	\qed
\end{itemize}
\end{proof}

%

\begin{proof}[of Lemma~\ref{lem:minterm-enabled}]
\hfill
\begin{enumerate}
	\item $p \tr{\varphi}{\checkreg{r}} q$ is enabled if and only if $v(r) \in \den{\varphi}$. Since both $\theta_r$ and $\varphi$ are minterms, this holds if and only if $\theta_r = \varphi$. 
	\item Suppose the transition is enabled, that is, for any $(p,v)$ such that $v \models \theta$ there is $(p,v) \xrightarrow{a} (q,v[r \mapsto a])$ in $\clts(\sra)$, with $a \in \den{\varphi} \setminus \img(v)$. Since $v$ is injective on non-empty registers and $\varphi \neq \bot$, by definition of minterm, $v$ picks $\card{\{ r \in R \mid \theta_r = \varphi \}}$ distinct elements from $\den{\varphi}$. Therefore, for $a$ to exist, $\den{\varphi}$ must have at least $\card{\{ r \in R \mid \theta_r = \varphi \}} + 1$ elements.

	Viceversa, suppose $\card{\den{\varphi}} > \card{\{ r \in R \mid \theta_r = \varphi \}} = k$, and take any $v \models \theta$. By injectivity of $v$ on non-empty registers, we have $\card{\img(v) \cap \den{\varphi}} = k$. Since $\card{\den{\varphi}} > k$, there exists $a \in \den{\varphi} \setminus \img(v)$, therefore there is $(p,v) \xrightarrow{a} (p,v[r \mapsto a])$ in $\clts(\sra)$.
\end{enumerate}	
\qed
\end{proof}

\begin{proof}[of Proposition~\ref{prop:norm-bisim}]
We will prove that 
\[
	\ltssim = \{ ( (p,v) , (\ns{p}{\theta},v) \mid v \models \theta \} 
\]
is a bisimulation. We first prove that it is a simulation. Consider $(p,v) \in \clts(\sra)$ and any transition $(p,v) \xrightarrow{a} (q,w)$. We will prove that this transition can be simulated by one from $((p,v),\theta)$ such that $v \models \theta$. We proceed by cases on which transition in $\Delta$ has generated $(p,v) \xrightarrow{a} (q,w)$:
\begin{itemize}
	\item $p \tr{\varphi}{\checkreg{r}} q$: then $w = v$, $a \in \den{\varphi}$ and $v(r) = a$ and $a \in \den{\varphi}$. By $v \models \theta$ we also have $a \in \den{\theta_r}$, therefore $\varphi \inm \theta_r$, which implies $\ns{p}{\theta} \tr{\theta_r}{\checkreg{r}} (q,\theta) \in \nrm{\Delta}$. By Lemma~\ref{lem:minterm-enabled}(1) this transition is enabled, therefore there is $(\ns{p}{\theta},v) \xrightarrow{a} ((q,\theta),v)$ in $\clts(\nrm{\sra})$, which is the required simulating transition. In fact, $v \models \theta$ implies $(q,v)\ltssim((q,\theta),v)$.
	\item $p \tr{\varphi}{\fresh{r}} q$: then $w = v[r \mapsto a]$ and $a \in \den{\varphi} \setminus \img(v)$. Take the minterm $\psi \in \mint(\sra)$ such that $a \in \den{\psi}$ (there is only one, by definition of minterm). We will prove that $\ns{p}{\theta} \tr{\psi}{\fresh{r}} (q,\theta[r \mapsto \psi]) \in \nrm{\Delta}$. Clearly we have $\psi \inm \varphi$, as $a \in \den{\varphi}$. We have to check $\card{\den{\psi}} > \card{\{ r \in R \mid \theta_r = \psi\}}$: the proof is the same as Lemma~\ref{lem:minterm-enabled}(2), showing by contradiction that $(p,v) \xrightarrow{a} (q,v[r \mapsto a])$ such that $a \in \den{\psi}$ cannot be generated from $p \tr{\varphi}{\fresh{r}} q$ whenever $\card{\den{\psi}} \leq \card{\{ r \in R \mid \theta_r = \psi\}}$. Therefore $\ns{p}{\theta} \tr{\psi}{\fresh{r}} (q,\theta[r \mapsto \psi]) \in \nrm{\Delta}$ and, by Lemma~\ref{lem:minterm-enabled}(2), it is enabled, i.e., there is $(\ns{p}{\theta},v) \xrightarrow{b} ((p,\theta[r \mapsto \psi]),v[r \mapsto b])$ such that $b \in \den{\psi} \setminus \img(v)$. Notice that there is such a transition for any $b \in \den{\psi} \setminus \img(v)$, in particular for $b = a$ we get the required simulating transition. In fact, because of the assumptions $v \models \theta$ and $a \in \den{\psi}$, we have $v[r \mapsto a] \models \theta[r \mapsto \psi]$ from which, by definition of $\ltssim$, we get $(p,v[r \mapsto a])\ltssim((p,\theta[r \mapsto \psi]),v[r \mapsto a])$. 	
\end{itemize}
We now prove that $\ltssim^{-1}$ is a simulation. Consider $(\ns{p}{\theta},v) \in \clts(\nrm{\sra})$ such that $v \models \theta$ and any transition $(\ns{p}{\theta},v) \xrightarrow{a} ((q,\gamma),w)$. We will prove that this transition can be simulated by one from $(p,v)$. We proceed by cases on which transition in $\nrm{\Delta}$ has generated $(\ns{p}{\theta},v) \xrightarrow{a} ((q,\gamma),w)$:
\begin{itemize}
	\item $\ns{p}{\theta} \tr{\theta_r}{\checkreg{r}} (q,\theta)$: then $\theta = \gamma$, $w = v$, $a \in \den{\theta_r}$, $v(r) = a$. By definition of $\nrm{\Delta}$, there is $p \tr{\varphi}{\checkreg{r}} q \in \Delta$ such that $\varphi \inm \theta_r$, which implies $\den{\theta_r} \subseteq \den{\varphi}$. Therefore $a \in \den{\varphi}$, and this transition from $\Delta$ instantiates to $(p,v) \xrightarrow{a} (q,v) \in \clts(\sra)$. This is the required simulating transition, because $v \models \theta$ implies $((q,\theta),v) \ltssim^{-1} (q,v)$.
	\item $\ns{p}{\theta} \tr{\psi}{\fresh{r}} (q,\theta[r \mapsto \psi])$: then $\gamma = \theta[r \mapsto \psi]$, $w = v[r \mapsto a]$ and $a \in \den{\psi} \setminus \img(v)$. By definition of $\nrm{\Delta}$, there is $p \tr{\varphi}{\fresh{r}} q \in \Delta$ such that $\varphi \inm \psi$. Therefore we have $\den{\psi} \subseteq \den{\varphi}$, which implies $a \in \den{\varphi} \setminus \img(v)$, so there is $(p,v) \xrightarrow{a} (q,v[r \mapsto a])$ in $\clts(\sra)$ corresponding to $p \tr{\varphi}{\fresh{r}} q$. This is the required simulating transition. In fact, because of the assumptions $v \models \theta$ and $a \in \den{\psi}$, we have $v[r \mapsto a] \models \theta[r \mapsto \psi]$ from which, by definition of $\ltssim$, we get $((p,\theta[r \mapsto \psi]),v[r \mapsto a]) \ltssim^{-1} (p,v[r \mapsto a])$.	
\end{itemize}
\qed
\end{proof}

\begin{proof}[of Lemma~\ref{lem:path-reach}]
For the first part of the claim, we shall prove that there is a path $\ns{q_0}{\theta_0} \tr{\varphi_1}{\ell_1} \ns{q_1}{\theta_1} \tr{\varphi_2}{\ell_2} \dots \tr{\varphi_n}{\ell_n} \ns{q_n}{\theta_n}$ in $\nrm{\sra}$ if and only if there is a run $(\ns{q_0}{\theta_0}, v_0) \xrightarrow{a_1} (\ns{q_1}{\theta_1},v_1) \xrightarrow{a_2} \dots \xrightarrow{a_n} (\ns{q_n}{\theta_n},v_n)$ in $\clts(\nrm{\sra})$ such that $v_n \models \theta_n$.
The right-to-left part is obvious, because every run in $\clts(\nrm{\sra})$ is instance of a path of $\nrm{\sra}$.
For the other direction, we will proceed by the length $n$ of the paths.
%
It is obvious for $n = 0$. For $n > 0$, suppose it holds for a path ending in $\ns{q_n}{\theta_n}$, so there is a run ending in $(\ns{q_n}{\theta_n},v_n)$ for all $v_n \models \theta_n$. Suppose there is 
\begin{equation}
	\ns{q_n}{\theta_n} \tr{\varphi_{n+1}}{\ell_{n+1}} \ns{q_{n+1}}{\theta_{n+1}} \in \nrm{\Delta}.
	\label{eq:reach-tr}
\end{equation}	
Because $v_n \models \theta_n$ by induction hypothesis and \eqref{eq:reach-tr} is enabled by $\theta_n$, there is $(\ns{q_n}{\theta_n},v_n) \xrightarrow{a} (\ns{q_{n+1}}{\theta_{n+1}},v_{n+1})$ in $\clts(\nrm{\sra})$. We either have $\theta_{n+1} = \theta_n$ and $v_{n+1} = v_n$, if $\ell_{n+1} \in R$, or $\theta_{n+1} = \theta_n[r \mapsto \varphi_{n+1}]$ and $v_{n+1} = v_n[a \mapsto r]$, if $\ell_{n+1} = \fresh{r}$, with $a \in \den{\varphi_{n+1}}$. In both cases we have $v_{n+1} \models \theta_{n+1}$, which concludes this part of the claim.

For the second part, we will show that if there is a path of length $n$ reaching $(\ns{p}{\theta},v)$ then there is one of the same length reaching $(\ns{p}{\theta},w)$, for all $w \models \theta$. We proceed again by induction on $n$. It is obvious for $n=0$, because only $(\ns{q_0}{\theta_0},v_0)$ is reachable. Suppose it holds for $n > 0$, and suppose $(\ns{q_{n+1}}{\theta_{n+1}},v_{n+1})$ can be reached via a path of length $n+1$. Then we have to show that all $(\ns{q_{n+1}}{\theta_{n+1}},w_{n+1})$ such that $w_n \models \theta_n$ are reachable. Suppose $(\ns{q_n}{\theta_n},v_n) \xrightarrow{a} (\ns{q_{n+1}}{\theta_{n+1}},v_{n+1})$ is the last transition in the run to $(\ns{q_{n+1}}{\theta_{n+1}},v_{n+1})$. Then this is instance of a transition of $\nrm{\sra}$ of the form \eqref{eq:reach-tr}. We proceed by cases on $\ell_{n+1}$:
\begin{itemize}
	\item if $\ell_{n+1} = \checkreg{r}$, then $\theta_{n+1} = \theta_n$ and $v_{n+1} = v_n$. By inductive hypothesis, we have that all $(\ns{q_n}{\theta_n},w_n)$ such that $w_n \models \theta_n$ are reachable. Since \eqref{eq:reach-tr} is enabled by $\theta_n$, it follows that there are $(\ns{q_n}{\theta_n},w_n) \xrightarrow{w_n(r)}  (\ns{q_{n+1}}{\theta_n},w_n)$, for all $w_n \models \theta_n$, which concludes the proof for this case.
	\item if $\ell_{n+1} = \fresh{r}$, then $\theta_{n+1} = \theta_n[r \mapsto \varphi_{n+1}]$ and $v_{n+1}(r) = a$. Consider any $w_{n+1} \models \theta_{n+1}$, with $w_{n+1} \neq v_{n+1}$. We will show that there must be $b \in \den{(\theta_n)_r}$ such that $w_{n+1}[r \mapsto b]$ is injective, which implies $w_{n+1}[r \mapsto b] \models \theta_n$. The claim would then follow: in fact, $(\ns{q_n}{\theta_n},w_{n+1}[r \mapsto b])$ would be reachable by the inductive hypothesis, and since \eqref{eq:reach-tr} is enabled by $\theta_n$, there is an instance of \eqref{eq:reach-tr} going to $(\ns{q_{n+1}}{\theta_{n+1}},w_{n+1})$. Suppose there is no such $b$. Then $\den{(\theta_n)_r}$ is already completely contained in the image of $\fres{w_{n+1}}{(R \setminus \{r\})}$. Since $w_{n+1}$ is injective, and $\theta_n$ and $\theta_{n+1}$ are made of minterms, we have 
	\[
		\card{\{ r' \in R \setminus \{r\} \mid (\theta_{n+1})_{r'} =  (\theta_n)_r \} } = \card{\den{(\theta_n)_r}}
	\]
	Now, $(\theta_n)_{r'} = (\theta_{n+1})_{r'}$ for $r' \in R \setminus \{r\}$, so from the equation above we get that $\theta_n$ maps ${\den{(\theta_n)_r}} + 1$ registers to $(\theta_n)_r$. This contradicts reachability of $(\ns{q_n}{\theta_n},v_n)$, because there cannot be any injective $v_n$ such that $v_n \models \theta_n$.\qed
\end{itemize}
\end{proof}

\begin{proof}[of Proposition~\ref{prop:red-det}]
For the right-to-left implication, suppose $\nrm{\sra}$ is not deterministic. Then there are two reachable transitions $(\ns{p}{\theta},v) \xrightarrow{a} (\ns{q_1}{\theta_1},v_1)$ and $(\ns{p}{\theta},v) \xrightarrow{a}(\ns{q_2}{\theta_2},v_2)$ in $\clts(\nrm{\sra})$ such that $(\ns{q_1}{\theta_1},v_1) \neq (\ns{q_2}{\theta_2},v_2)$. Suppose these two transitions are respectively instances of $\ns{p}{\theta} \tr{\varphi_1}{\ell_1} \ns{q_1}{\theta_1}$ and $\ns{p}{\theta} \tr{\varphi_2}{\ell_2} \ns{q_2}{\theta_2}$. Then we have the following cases, corresponding to the ones in the statement:
\begin{enumerate}
	\item Suppose $v_1 = v_2$, then we must have $\ell_1 = \ell_2$ and $\ns{q_1}{\theta_1} \neq \ns{q_2}{\theta_2}$, because we assumed $(\ns{q_1}{\theta_1},v_1) \neq (\ns{q_2}{\theta_2},v_2)$. Since $a \in \den{\varphi_1 \land \varphi_2}$ and $\varphi_1,\varphi_2$ are minterms, we obtain $\varphi_1 = \varphi_2$.
	\item Suppose $v_1 \neq v_2$, then we must have $\ell_1 = \fresh{r}$ and $\ell_2 = \fresh{s}$, for $r \neq s$, because assignments are only changed by transitions reading fresh symbols. Again, we have $a \in \den{\varphi_1 \land \varphi_2}$, thus $\varphi_1 = \varphi_2$.
\end{enumerate}

For the other direction, we assume that there are two reachable transitions $\ns{p}{\theta} \tr{\varphi_1}{\ell_1} \ns{q_1}{\theta_1}, \ns{p}{\theta} \tr{\varphi_2}{\ell_2}\ns{q_2}{\theta_2} \in \nrm{\Delta}$ that violate the condition, and we show that $\nrm{\sra}$ is not deterministic. We proceed by cases:
\begin{itemize}
	\item  $\ell_1 = \ell_2$, $q_1 \neq q_2$ and $\varphi_1 = \varphi_2$. Then $\theta_1 = \theta_2$. Take a reachable configuration $(\ns{p}{\theta},v)$, which exists because $\ns{p}{\theta}$ is reachable, by Proposition~\ref{prop:norm-bisim}. Then the two transitions can be instantiated to $(\ns{p}{\theta},v) \xrightarrow{v(r)} (\ns{q_1}{\theta},v)$ and $(\ns{p}{\theta},v) \xrightarrow{v(r)} (\ns{q_2}{\theta},v)$. Nondeterminism follows from $q_1 \neq q_2$.
	\item $\ell_1 = \fresh{r}$, $\ell_2 = \fresh{s}$, $r \neq s$ and $\varphi_1 = \varphi_2$. Take a reachable configuration $(\ns{p}{\theta},v)$, then since the two transitions are enabled by construction, there is $a \in \den{\varphi_1} \setminus \img(v)$. Therefore the two transitions can be instantiated to $(\ns{p}{\theta},v) \xrightarrow{a} (\ns{q_1}{\theta_1},v[r \mapsto a])$ and $(\ns{p}{\theta},v) \xrightarrow{a} (\ns{q_2}{\theta_2},v[s \mapsto a])$. By $r \neq s$, we have $v[r \mapsto a] \neq v[s \mapsto a]$, from which nondeterminism follows.
\end{itemize}
\qed
\end{proof}

\subsection{Proof of Theorem~\ref{thm:nrmsim-sound-complete}}
Given any two \SRAs $\sra_1$ and $\sra_2$, by Theorem~\ref{prop:norm-bisim} we have $\sra_1 \prec \sra_2$ if and only if $\nrm{\sra_1} \prec \nrm{\sra_2}$. We will now show that $\nrm{\sra_1} \prec \nrm{\sra_2}$ if and only if $\sra_1 \nrmsim \sra_2$. We show the two directions separately in the following lemmata.

\begin{lemma}
Let $\srasim$ be a $\textsf{N}$-simulation such that 
\[
(\nrm{q_{01}},\nrm{q_{02}},v_{01} \bowtie v_{02}) \in \srasim.
\] 
Then the following relation:
\begin{align*}
	\srasim' = \{ ((\ns{p_1}{\theta_1},v_1),(\ns{p_2}{\theta_2},v_2)) \mid &(\ns{p_1}{\theta_1},\ns{p_2}{\theta_2},\sigma) \in \srasim, \\
	&  v_1 \models \theta_1, v_2 \models \theta_2, v_1,v_2 \models \sigma \}
\end{align*}	
is a simulation on $\nrm{\sra_1}$ and $\nrm{\sra_2}$ such that $((\nrm{q_{01}},v_{01}),(\nrm{q_{02}},v_{02})) \in \srasim'$.
\end{lemma}

\begin{proof} 
First of all, by definition of $\srasim'$, $(\nrm{q_{01}},\nrm{q_{02}},v_{01} \bowtie v_{02}) \in \srasim$ implies $((\nrm{q_{01}},v_0),(\nrm{q_{02}},\sigma_0)) \in \srasim'$. The other conditions follow by definition of $\theta_{01},\theta_{02}$ and $v_{01} \bowtie v_{02}$.

We will now prove that $\srasim'$ is a simulation.
Suppose $((\ns{p_1}{\theta_1},v_1),(\ns{p_2}{\theta_2},v_2)) \in \srasim'$, hence $(\ns{p_1}{\theta_1},\ns{p_2}{\theta_2},\sigma) \in \srasim$. By Definition~\ref{def:sym-bisim}, we must have that $\ns{p_1}{\theta_1} \in \nrm{F_1}$ implies $\ns{p_2}{\theta_2} \in \nrm{F_2}$. Now, suppose $(\ns{p_1}{\theta_1},v_1) \xrightarrow{a} (\ns{q_1}{\theta_1'},v_1')$. We have to prove that there is $(\ns{p_2}{\theta_2},v_2) \xrightarrow{a} (\ns{q_2}{\theta_2'},v_2')$ such that $((\ns{q_1}{\theta_1'},v_1'),(\ns{q_2}{\theta_2'},v_2')) \in \srasim'$. Let $(\ns{p_1}{\theta_1},v_1) \xrightarrow{a} (\ns{q_1}{\theta_1'},v_1')$ be instance of $\ns{p_1}{\theta_1} \tr{\varphi_1}{\ell_1} \ns{q_1}{\theta_1'}$. We proceed by cases on $\ell_1$:
\begin{enumerate}
	\item $\ell_1 = \checkreg{r}$: Then $\theta_1' = \theta_1$, $v_1' = v_1$, $a = v_1(r)$ and we have two cases. 
	\begin{enumerate}
		\item If $r \in \dom(\sigma)$, then by definition of $\srasim$ there is $\ns{p_2}{\theta_2} \tr{\varphi_1}{\checkreg{\sigma(r)}} \ns{q_2}{\theta_2} \in \nrm{\Delta_2}$ such that $(\ns{q_1}{\theta_1}, \ns{q_2}{\theta_2},\sigma) \in \srasim$. This transition is enabled by $\theta_2$, so by $v_2 \models \theta_2$ and $a = v_2(\sigma(r))$, by definition of $\sigma$, we get the existence of $(\ns{p_2}{\theta_2},v_2) \xrightarrow{a} (\ns{p_2}{\theta_2},v_2)$ in $\clts(\nrm{\sra_2})$. This is the required simulating transition. In fact, from $(\ns{q_1}{\theta_1}, \ns{q_2}{\theta_2},\sigma) \in \srasim$, and the assumptions $v_1 \models \theta_1$ and $v_2 \models \theta_2$, we obtain $((\ns{q_1}{\theta_1},v_1),(\ns{q_2}{\theta_2},v_2)) \in \srasim'$, by definition of $\srasim'$.
		\item If $r \notin \dom(\sigma)$, then by definition of $\srasim$ there is $\ns{p_2}{\theta_2} \tr{\varphi_1}{\fresh{s}} \ns{q_2}{\theta_2[r \mapsto \varphi_1]} \in \nrm{\Delta_2}$ such that $(\ns{q_1}{\theta_1}, \ns{q_2}{\theta_2[r \mapsto \varphi_1]},\sigma[r \mapsto s]) \in \srasim$. Since $r \notin \dom(\sigma)$ and $a \in \den{\varphi_1}$, we have $a \in \den{\varphi_1} \setminus \img(v_2)$. Therefore there is $(\ns{p_2}{\theta_2},v_2) \xrightarrow{a} (\ns{q_2}{\theta_2[r \mapsto \varphi_1]},v_2[s \mapsto a]) \in \nrm{\Delta_2}$. This is the required simulating transition. In fact, we have $v_1 \models \theta_1$ and, from $v_2 \models \theta_2$ and $a \in \den{\varphi_1}$, we get $v_2[s \mapsto a] \models \theta_2[s \mapsto \varphi_1]$. Moreover, we have $(\ns{q_1}{\theta_1}, \ns{q_2}{\theta_2[r \mapsto \varphi_1]},\sigma[r \mapsto s]) \in \srasim$, and $\sigma[r \mapsto s]$ encodes the new equality $v_1(r) = v_2[r \mapsto v_1(r)](s)$. Therefore, by definition of $\srasim'$, we obtain $((\ns{q_1}{\theta_1},v_1),(\ns{q_2}{\theta_2[s \mapsto \varphi_1]},v_2[s \mapsto a])) \in \srasim'$.
	\end{enumerate}
	\medskip
	\item $\ell_1 = \fresh{r}$: Then $\theta_1' = \theta[r \mapsto \varphi_1]$, $v_1' = v_1[r \mapsto a]$, $a \in \den{\varphi_1} \setminus \img(v_1)$, and we have two cases:
	\begin{enumerate}
		\item There is $s \in R_2$ such that $v_2(s) = a$. Since $a \notin \img(v_1)$, we have $s \in R_2 \setminus \dom(\sigma)$, so by Definition~\ref{def:red-sym-bisim} there is $\ns{p_2}{\theta_2} \tr{\varphi_1}{\checkreg{s}} \ns{q_2}{\theta_2} \in \nrm{\Delta_2}$ such that $(\ns{q_1}{\theta_1[r \mapsto \varphi_1]}, \ns{q_2}{\theta_2},\sigma[r \mapsto s]) \in \srasim$, which instantiates to $(\ns{p_2}{\theta_2},v_2) \xrightarrow{a} (\ns{q_2}{\theta_2},v_2)$ in $\clts(\nrm{\sra_2})$. The argument is then similar to the point 1.(b).
		\item No $s \in R_2$ is such that $v_2(s) = a$. Then we must have $\regab(\theta_1,\varphi_1) + \regab(\theta_2,\varphi_1) < \card{\den{\varphi_1}}$. In fact, if this is not the case, we would have 
		\[
			\card{(\img(v_1) \cup \img(v_2)) \cap \den{\varphi_1}} = \card{\den{\varphi_1}}
		\]
		by injectivity of $v_1$ and $v_2$ and $v_1 \models \theta_1, v_2 \models \theta_2$. Therefore there would be no $a \in \den{\varphi_1} \setminus (\img(v_1) \cup \img(v_2))$, a contradiction.
		 Hence, by definition of $\srasim$, there is $\ns{p_2}{\theta_2} \tr{\varphi_1}{\fresh{s}} \ns{q_2}{\theta_2[s \mapsto \varphi_1]} \in \nrm{\Delta_2}$ such that 
		\begin{equation}
			(\ns{q_1}{\theta_1[r \mapsto \varphi_1]}, \ns{q_2}{\theta_2[s \mapsto \varphi_1]},\sigma[r \mapsto s]) \in \srasim.\label{eq:rb-2b}
		\end{equation} 
		Because $a \in \den{\varphi_1} \setminus \img(v_2)$, this transition instantiates to $(\ns{p_2}{\theta_2},v_2) \xrightarrow{a} (\ns{p_2}{\theta_2[r \mapsto \varphi_1]},v_2[s \mapsto a])$ in $\clts(\nrm{\sra}_2)$. Since $v_i \models \theta_i$, $i=1,2$, $a \in \den{\varphi_1}$ and $v_1,v_2 \models \sigma$, we have $v_1[r \mapsto a] \models \theta_2[r \mapsto \varphi_1]$, $v_2[s \mapsto a] \models \theta_2[s \mapsto \varphi_1]$, and $v_1[r \mapsto a],v_2[s \mapsto a] \models \sigma[r \mapsto s]$. Therefore, by definition of $\srasim'$, \eqref{eq:rb-2b} implies \[
			((\ns{q_1}{\theta_1[r \mapsto \varphi_1]},v_1[r \mapsto a]), (\ns{q_2}{\theta_2[s \mapsto \varphi_1]},v_2[s \mapsto a]) \in \srasim'.
		\]
	\end{enumerate}
\end{enumerate}
\end{proof}

\begin{lemma}
Let $\srasim$ be a simulation on $\nrm{\sra_1}$ and $\nrm{\sra_2}$ such that 
\[
	((\nrm{q_{01}},v_{01}),(\nrm{q_{02}},v_{02}) \in \srasim.
\] 
Then the following relation:
\begin{align*}
	\srasim' = \{ (\ns{p_1}{\theta_1},\ns{p_2}{\theta_2},\sigma) \mid &\exists v_1 \models \theta_1 ,v_2 \models \theta_2 : ((\ns{p_1}{\theta_1},v_1),(\ns{p_2}{\theta_2},v_2)) \in  \srasim, \\
	&\sigma = v_1 \bowtie v_2 \} 
\end{align*}
is a \textsf{N}-simulation on $\nrm{\sra_1}$ and $\nrm{\sra_2}$ such that $((\nrm{q_{01}},v_0),(\nrm{q_{02}},\sigma_0)) \in \srasim'$
\end{lemma}
\begin{proof}
First of all, $((\nrm{q_{01}},v_0),(\nrm{q_{02}},\sigma_0)) \in \srasim'$ clearly follows from 
\[
	((\nrm{q_{01}},v_{01}) , (\nrm{q_{02}},v_{02})) \in \srasim.
\]

We will now prove that $\srasim'$ is a $\textsf{N}$-simulation. Suppose $(\ns{p_1}{\theta_1},\ns{p_2}{\theta_2},\sigma) \in \srasim'$. Then, by definition of $\srasim'$, there are $((\ns{p_1}{\theta_1},v_1),(\ns{p_2}{\theta_2},v_2)) \in  \srasim$ such that $v_1 \models \theta_1 ,v_2 \models \theta_2$. Since $\srasim$ is a simulation, we must have that $\ns{p_1}{\theta_1} \in \nrm{F_1}$ implies $\ns{p_1}{\theta_1} \in \nrm{F_2}$. Now, suppose $\ns{p_1}{\theta_1} \tr{\varphi_1}{\ell_1} \ns{q_1}{\theta_1'} \in \nrm{\Delta_1}$. We have to exhibit transitions from $\ns{p_2}{\theta_2}$ that match the definition of normalised symbolic bisimulation:
\begin{enumerate}
	\item $\ell_1 = \checkreg{r}$: then $\theta_1' = \theta_1$. Since $\ns{p_1}{\theta_1} \tr{\varphi_1}{\checkreg{r}} \ns{q_1}{\theta_1}$ is enabled by $\theta_1$, by construction, and $v_1 \models \theta_1$, there is $(\ns{p_1}{\theta_1},v_1) \xrightarrow{a} (\ns{q_1}{\theta_1},v_1)$ in $\clts(\nrm{\sra_1})$. Since $\srasim$ is a bisimulation, there is $(\ns{p_2}{\theta_2},v_2) \xrightarrow{a} (\ns{q_1}{\theta_2'},v_2')$ in $\clts(\nrm{\sra_2})$. We have two cases:
	\begin{enumerate}
		\item $r \in \dom(\sigma)$, then $v_2(\sigma(r)) = a$, and the transition in $\clts(\sra_2)$ must be instance of some $\ns{p_2}{\theta_2} \tr{\varphi_2}{\checkreg{r}} \ns{q_2}{\theta_2} \in \nrm{\Delta_2}$. Now, since $a \in \den{\varphi_1} \cap \den{\varphi_2}$, and $\varphi_1$ and $\varphi_2$ are minterms (recall that $\nrm{\sra_1}$ and $\nrm{\sra_2}$ are defined over the same set of minterms), we must have $\varphi_2 = \varphi_1$. Therefore the given transition of $\nrm{\sra_2}$ is the required simulating one. In fact, because $\srasim$ is a bisimulation, we have $((\ns{q_1}{\theta_1},v_1),(\ns{q_2}{\theta_2},v_2)) \in \srasim$, which implies $(\ns{q_1}{\theta_1},\ns{q_2}{\theta_2},\sigma) \in \srasim'$, by definition of $\srasim'$. 
		\item $r \notin \dom(\sigma)$, then $a \notin \img(v_2)$, and the transition in $\clts(\sra_2)$ must be instance of some $\ns{p_2}{\theta_2} \tr{\varphi_2}{\fresh{s}} \ns{q_2}{\theta_2[s \mapsto \varphi_2]} \in \nrm{\Delta_2}$. By the same reasoning as above, we must have $\varphi_1 = \varphi_2$. Because $\srasim$ is a bisimulation, we have $((\ns{q_1}{\theta_1},v_1),(\ns{q_2}{\theta_2[s \mapsto \varphi_2]},v_2[s \mapsto a])) \in \srasim$, which implies $(\ns{q_1}{\theta_1},\ns{q_2}{\theta_2[r \mapsto \varphi_1]},v_1 \bowtie v_2[s \mapsto a]) \in \srasim'$, by definition of $\srasim'$. 
	\end{enumerate}
	\item $\ell_1 = \fresh{r}$: then $\theta_1' = \theta[r \mapsto \varphi_1]$. Since $\ns{p_1}{\theta_1} \tr{\varphi_1}{\checkreg{r}} \ns{q_1}{\theta[r \mapsto \varphi_1]}$ is enabled by $\theta[r \mapsto \varphi_1]$, by construction, and $v_1 \models \theta_1$, there is $(\ns{p_1}{\theta_1},v_1) \xrightarrow{a} (\ns{q_1}{\theta[r \mapsto \varphi_1]},v_1[r \mapsto a])$ in $\clts(\nrm{\sra_1})$, for all $a \in \den{\varphi_1} \setminus \img(v_1)$. The values of $a$ can be of two kinds:
	\begin{enumerate}
		\item $a = v_2(s)$, for some $s$. This holds if and only if $s \in R_2 \setminus \dom(\sigma)$ and $\varphi_1 = (\theta_2)_s$. In fact, $s \in R_2 \setminus \dom(\sigma)$ if and only if $v_2(s) \notin \img(v_1)$, by definition of $\sigma$, and $\varphi_1 = (\theta_2)_s$ if and only if $v_2(s) \in \den{\varphi_1}$, because the two predicates are minterms. Therefore the condition on $s$ above is equivalent to $v_2(s) \in \den{\varphi_1} \setminus \img(v_1)$.
		
		Because $\srasim$ is a simulation, there is a transition $(\ns{p_2}{\theta_2},v_2) \xrightarrow{a} (\ns{q_2}{\theta_2},v_2)$ in $\clts(\nrm{\sra_2})$. This transition must be instance of some $\ns{p_2}{\theta_2} \tr{\varphi_1}{\checkreg{s}} \ns{q_2}{\theta_2} \in \nrm{\Delta_2}$. This is the required simulating transition. In fact, we have $((\ns{q_1}{\theta_1[r \mapsto \varphi_1]},v_1[r \mapsto a]),(\ns{q_2}{\theta_2},v_2)) \in \srasim$, which implies 
		\[
			(\ns{q_1}{\theta_1},\ns{q_2}{\theta_2},v_1[r \mapsto a] \bowtie v_2) \in \srasim'
		\] 
		Notice that $v_1[r \mapsto a] \bowtie v_2 = \sigma[r \mapsto s]$ follows from the assumption $\sigma = v_1 \bowtie v_2$, and from $v_1[r \mapsto a](r) = v_2(s)$.
		\item $a \notin \img(v_2)$. Then the transition above must be an instance of $\ns{p_2}{\theta_2} \tr{\varphi_1}{\fresh{s}} \ns{q_2}{\theta_2[r \mapsto \varphi_1]}$. The argument proceeds similarly as the point above.
	\end{enumerate}	\qed
\end{enumerate}
\end{proof}

\begin{proof}[of Proposition~\ref{prop:sra-nrm-det}]

Suppose $\sra$ is not deterministic, then there is a reachable configuration $(p,v)$ in $\clts(\sra)$ and two transitions $(p,v) \xrightarrow{a} (q_1,v_1)$ and $(p,v) \xrightarrow{a} (q_2,v_2)$, with $(q_1,v_1) \neq (q_2,v_2)$. By Proposition~\ref{prop:norm-bisim}, these transitions exist if an only if there are two transitions $(\ns{p}{\theta},v) \xrightarrow{a} (\ns{q_1}{\theta_1},v_1)$, $(\ns{p}{\theta},v) \xrightarrow{a} (\ns{q_2}{\theta_2},v_2)$ in $\clts(\nrm{\sra})$ i.e., if and only if $\nrm{\sra}$ is not deterministic.
\qed
\end{proof}

\section{Additional results}

\newcommand{\sink}{\mathsf{sink}}
\begin{proposition}
Given a deterministic single-valued SRA $\sra$, let $\mathsf{pred} \colon Q \times (R \cup \{\bullet \}) \to \Psi$ be the function
\[
	\mathsf{pred}(p,x) = 
	\begin{cases}
		\bigvee \{ \varphi \mid p \tr{\varphi}{x} q \in \Delta \} & x \in R \\
		\bigvee \{ \varphi \mid p \tr{\varphi}{\fresh{r}} q \in \Delta \} & x = \bullet
	\end{cases}
\]
where $\bigvee \emptyset = \bot$.
Let $\sra' = (R, Q \cup \{\sink\}, q_0,v_0,,F,\Delta')$ be given by
\begin{align*}
	\Delta' &= \Delta \cup \Delta_{\neg} \cup \Delta_{\sink}
	\\
	\Delta_{\neg} &= \bigcup_{p \in Q} \{ p \tr{\neg \mathsf{pred}(p,r)}{r} \sink \mid r \in R \} \cup \{ p \tr{\neg \mathsf{pred}(p,\bullet)}{\fresh{\tilde{r}}} \sink \} \\
	\Delta_{\sink} &= \{ \sink \tr{\top}{r} \sink \mid r \in R \} \cup \{ \sink \tr{\top}{\fresh{\tilde{r}}} \sink \}  
\end{align*}
where $\tilde{r} \in R$ is a chosen register. Then $\sra'$ is complete and deterministic, and $\lang(\sra') = \lang(\sra)$.
\label{prop:sra-complete}
\end{proposition}
\begin{proof}
We first show that $\sra'$ is complete, i.e, that for any configuration $(p,v)$ of $\sra'$ and any $a \in \els$ there is a transition $(p,v) \xrightarrow{a} (q,w)$ in $\clts(\sra')$. This is clearly true for $p = \sink$, as its outgoing transitions can read anything. If $p \neq \sink$, then either there is $r \in R$ such that $v(r) = $, or $a \notin \img(v)$. In the first case, either $a \in \den{\mathsf{pred}(p,r)}$, and so a transition in $\Delta$ can read it, or $a \in \neg \den{\mathsf{pred}(p,r)}$, and so a transition in $\Delta_{\neg}$ can read it. The case $a \notin \img(v)$ is analogous.

Notice that this case analysis also tells us that the transitions in $\Delta$, $\Delta_{\neg}$ and $\Delta_{\sink}$ are mutually exclusive, from which determinism follows. 

The last claim $\lang(\sra) = \lang(\sra')$ follows from the fact that all the additional transitions of $\sra'$ go to $\sink$, which is non-accepting.\qed
\end{proof}

\end{appendix}